\documentclass[10pt,twocolumn,twoside]{IEEEtran}

\usepackage[T1]{fontenc}
\usepackage{microtype}
\usepackage{amsmath,amsthm}
\allowdisplaybreaks[4]
\usepackage{amssymb}
\usepackage{wasysym}
\usepackage{framed}
\usepackage{mathtools}
\mathtoolsset{showonlyrefs=false,showmanualtags=true}
\usepackage{verbatim}
\usepackage{graphicx}
\usepackage[dvipsnames,svgnames,x11names]{xcolor}
\usepackage{url}

\usepackage{cite}
\usepackage{enumerate}
\usepackage{enumitem}
\setlist[itemize]{leftmargin=*}
\usepackage{epstopdf}
\usepackage{tikz}
\usepackage{tikz-3dplot}
\usepackage[colorlinks=false, pdfborder={0 0 0}]{hyperref}

\definecolor{crimsonred}{RGB}{153,0,0}		
\definecolor{darkcharcoal}{RGB}{25,25,25}		
\definecolor{charcoal}{RGB}{51,51,51}		
\definecolor{ash}{RGB}{100,100,100}			
\definecolor{paleblue}{RGB}{0,102,102}		
\definecolor{turtlegreen}{RGB}{51,153,0}	
\definecolor{paleale}{RGB}{204,204,102}		
\definecolor{lager}{RGB}{140,110,10}		
\definecolor{regal}{RGB}{90,0,120}			
\definecolor{jeans}{RGB}{20,30,150}			
\definecolor{ALERT}{RGB}{153,0,0}			
\definecolor{Alert}{RGB}{51,153,0}			
\definecolor{alert}{RGB}{140,110,10}		
\definecolor{charcoal}{RGB}{80,80,80}		
\definecolor{comment}{RGB}{51,51,51}		
\definecolor{Comment}{RGB}{100,100,100}		
\definecolor{COMMENT}{RGB}{80,20,120}		
\definecolor{highlight}{RGB}{20,30,150}		

\newcommand{\tb}{\textcolor{Black}}

\newcommand{\novec}{{}}
\renewcommand{\matrix}{{}}

\newcommand{\mL}{\mathcal{L}}
\newcommand{\mN}{\mathcal{N}}

\newcommand{\mG}{\mathcal{G}}
\newcommand{\mK}{\mathcal{K}}

\renewcommand{\star}{*}

\newtheorem{theorem}{Theorem}[section]
\newtheorem{lemma}[theorem]{Lemma}

\title{Mechanism design for resource allocation in networks with intergroup competition and intragroup sharing}

\author{Abhinav Sinha, \emph{Student Member, IEEE}  and  Achilleas Anastasopoulos, \emph{Senior Member, IEEE}
\thanks{This work was supported by the National Science Foundation by Grant CIF-1111061.}
\thanks{This paper has initially appeared on July 2013 as a technical report {arXiv:1307.2569} with last revision on March 2017.}
\thanks{The authors are with the
Department of Electrical Engineering and Computer Science,
University of Michigan, Ann Arbor, MI 48109, USA
(e-mail: \texttt{\{absi,anastas\}@umich.edu}).}
}

\begin{document}

	\maketitle
	
	
	\begin{abstract}
We consider a network where strategic agents, who are contesting for allocation of resources, are divided into fixed groups.
The network control protocol is such that within each group agents get to share the resource and across groups they contest for it.
A prototypical example is the allocation of data rate on a network with multicast/multirate architecture.
Compared to the unicast architecture (which is a special case), the multicast/multirate architecture can result in substantial bandwidth savings.
However, design of a market mechanism in such a scenario requires dealing with both private and public good problems as opposed to just private goods in unicast.

The mechanism proposed in this work ensures that social welfare maximizing allocation on such a network is realized at all Nash equilibria (NE) i.e., full implementation in NE. In addition it is individually rational, i.e., agents have an incentive to participate in the mechanism.
The mechanism, which is constructed in a quasi-systematic way starting from the dual of the centralized problem, has a number of useful properties.
Specifically, due to a novel allocation scheme, namely ``radial projection'', the proposed mechanism results in feasible allocation even off equilibrium. This is a practical necessity for any realistic mechanism since agents have to ``learn'' the NE through a dynamic process.
%
%
Finally, it is shown how strong budget balance at equilibrium can be achieved with a minimal increase in message space as an add-on to a weakly budget balanced mechanism.
	\end{abstract}

	

	\section{Introduction} \label{Intro}


	\subsection{Motivation}
	
	Design of mechanisms that fully implement Walrasian and Lindahl allocations in Nash equilibrium (NE) have been extensively studied in the literature, e.g. \cite{groves, hurwicz1979outcome, tian1989implementation, postlewaite1989feasible, hurwicz1994economic, maskin2002implementation}.
	These two ubiquitous examples present two different aspects in market design - private and public goods, respectively.
	Recently, generalizing the original contributions described above, a number of works addressed the problem of full implementation in NE of social utility maximization under linear inequality constraints~\tb{ \cite{basar,stoenescu06,rahuljain,kakhbodcorrection,bhattacharya2014,SiAn15b}.}
	In the area of computer/communication networks, the prototypical example of this problem is the allocation of rate (bandwidth) in a ``unicast'' network architecture on the Internet, as depicted in Fig.~\ref{fig:unicast}.
	\begin{figure}[htbp]
		\centering
\scalebox{0.7}{
\begin{tikzpicture}
\node[draw, color=white] (m1) at (9.5,2.05) { \textcolor{blue}{\scriptsize $ {x_{2}}  $} };
\node[draw, rounded corners] (t1) at (0.5,2.1) {\textcolor{black}{T1}};
\node[draw, rounded corners] (t2) at (0.5,0.7) {\textcolor{black}{T2}};
%
\node[draw, color=red, rounded corners] (r1) at (10.5,2.6) {\textcolor{black}{R1}};
\node[draw, color=white] (r1n) at (11.5,2.6) { \textcolor{red}{\footnotesize $ x_{1} $} };
\node[draw, color=blue, rounded corners] (r2) at (10.5,1.8) {\textcolor{black}{R2}};
\node[draw, color=white] (r2n) at (11.5,1.8) { \textcolor{blue}{\footnotesize $ x_{2} $} };
\node[draw, color=Alert, rounded corners] (r3) at (10.5,1) {\textcolor{black}{R3}};
\node[draw, color=white] (r3n) at (11.5,1) { \textcolor{Alert}{\footnotesize $ x_{3} $} };
\node[draw, color=orange, rounded corners] (r4) at (10.5,0.2) {\textcolor{black}{R4}};
\node[draw, color=white] (r4n) at (11.5,0.2) { \textcolor{orange}{\footnotesize $ x_{4} $} };
\node[draw] (n13) at (8.5,2) {C};
\node[draw] (n21) at (2.5,1) {A};
\node[draw] (n22) at (5.5,1) {B};
\draw (0,-0.3) rectangle (1,3) ;
\node at (0.5, 3.4) {Tx} ;
\draw (10,-0.3) rectangle (11,3) ;
\node at (10.5, 3.4) {Rx} ;
\draw[color=red, line width=0.5pt] [-] (n21) -- (n22) -- (n13) -- (r1);
\draw[color=blue, line width=0.5pt] [-] (n21) -- (n22) -- (n13) -- (r2);
\draw[color=Alert, line width=0.5pt] [-] (n21) -- (n22) -- (n13) -- (r3);
\draw[color=orange, line width=0.5pt] [-] (n21) -- (n22) -- (r4);
\draw[line width=1pt,color=red] [-] (t1) -- (1.5,1.55)
node[pos=1,above] {\small~~~~~~\textcolor{Red}{$ {x_1} $}~$ + $~\textcolor{blue}{$ {x_2} $}} ;
\draw[line width=1pt,color=red] [-] (1.5,1.55) -- (n21) ;
\draw[line width=1pt,color=Red] [-] (n22) -- (6.5,1.33);
\draw[line width=1pt,color=blue] [-] (6.5,1.33) -- (7.5,1.66);
\draw[line width=1pt,color=Alert] [-] (7.5,1.66) -- (n13);
\node at (7,2) {\small \textcolor{Red}{$ {x_1} $}~$ + $~\textcolor{blue}{$ {x_2} $}~$ + $~\textcolor{Alert}{$ {x_3} $}};
\draw[line width=1pt,color=Red] [-] (n21) -- (3.4,1)
node[pos=0.5,below] {\small $ {x_1} $};
\node at (3.4,0.75) {\small $ + $};
\draw[line width=1pt,color=blue] [-] (3.4,1) -- (4.1,1)
node[pos=0.5,below] {\small $ {x_2} $};
\node at (4.05,0.75) {\small $ + $};
\draw[line width=1pt,color=Alert] [-] (4.1,1) -- (4.6,1)
node[pos=0.5,below] {\small $ {x_3} $};
\node at (4.65,0.75) {\small $ + $};
\draw[line width=1pt,color=orange] [-] (4.6,1) -- (n22)
node[pos=0.5,below] {\small $ {x_{4}} $} ;
\draw[line width=1pt,color=orange] [-] (n22) -- (r4)
node[pos=0.5,below] {\scriptsize $ {x_{4}} $} ;
\draw[line width=1pt,color=blue] [-] (1.5,1.55) -- (n21) ;
\draw[line width=1pt,color=Alert] [-] (t2) -- (1.5,0.85) 
node[pos=1,below] {\small~~~\textcolor{Alert}{$ {x_3} $}~$ + $~\textcolor{orange}{$ {x_4} $}} ;
\draw[line width=1pt,color=orange] [-] (1.5,0.85) -- (n21) ;
\draw[line width=1pt,color=red] [-] (n13) -- (r1)
node[pos=0.5,above] {\scriptsize $ {x_{1}}  $};
\draw[line width=1pt,color=blue] [-] (n13) -- (r2) ;
\draw[line width=1pt,color=Alert] [-] (n13) -- (r3)
node[pos=0.5,below] {\scriptsize $ {x_{3}}  $};
\end{tikzpicture}
}
		\caption{\tb{Network with ``unicast'' architecture}}
		\label{fig:unicast}
\scalebox{0.7}{
\begin{tikzpicture}
\node[draw, color=white] (m1) at (4,0.5) {  $ \textcolor{blue}{\max\{x_{1},{x_{2}}\}} $ };
\node[draw, color=white] (m5) at (4,0) {  {\textcolor{black}{$ + $}}   \textcolor{Alert}{ $ \max\{{x_{3},x_{4}}\} $ } };
\node[draw, color=blue, rounded corners] (t1) at (0.5,2.1) {\textcolor{black}{T1}};
\node[draw, color=Alert , rounded corners] (t2) at (0.5,1) {\textcolor{black}{T2}};
%
%
\node[draw, color=blue, rounded corners] (r1) at (10.5,2.6) {\textcolor{black}{R1}};
\node[draw, color=white] (r1n) at (11.5,2.6) { \textcolor{blue}{\footnotesize $ x_{1} $} };
\node[draw, color=blue, rounded corners] (r2) at (10.5,1.8) {\textcolor{black}{R2}};
\node[draw, color=white] (r2n) at (11.5,1.8) { \textcolor{blue}{\footnotesize $ x_{2} $} };
\node[draw, color=Alert, rounded corners] (r3) at (10.5,1) {\textcolor{black}{R3}};
\node[draw, color=white] (r3n) at (11.5,1) { \textcolor{Alert}{\footnotesize $ x_{3} $} };
\node[draw, color=Alert, rounded corners] (r4) at (10.5,0.2) {\textcolor{black}{R4}};
\node[draw, color=white] (r4n) at (11.5,0.2) { \textcolor{Alert}{\footnotesize $ x_{4} $} };
\node[draw, color=white] (m2) at (6.5,2) { \textcolor{blue}{\small $ \max\{{x_{1},x_{2}}\} \: \textcolor{black}{+} \: \textcolor{Alert}{{x_{3}}} $} };
\node[draw, color=white] (m1) at (9.5,2.05) { \textcolor{blue}{\scriptsize $ {x_{2}}  $} }; 
\node[draw] (n13) at (8.5,2) {C};
\node[draw] (n21) at (2.5,1) {A};
\node[draw] (n22) at (5.5,1) {B};
\draw (0,-0.3) rectangle (1,3) ;
\node at (0.5, 3.4) {Tx} ;
\draw (10,-0.3) rectangle (11,3) ;
\node at (10.5, 3.4) {Rx} ;
\draw[color=blue, line width=0.5pt] [-] (n22) -- (n13) -- (r1);
\draw[color=blue, line width=0.5pt] [-] (n21) -- (n22) -- (n13) -- (r2);
\draw[color=Alert, line width=0.5pt] [-] (t2) -- (n21) -- (n22) -- (n13) -- (r3);
\draw[color=Alert, line width=0.5pt] [-] (t2) -- (n21) -- (n22) -- (r4);
\draw[line width=1pt,color=blue] [-] (n22) -- (7,1.5) ;
\draw[line width=1pt,color=Alert] [-] (7,1.5) -- (n13);
\draw[line width=1pt,color=blue] [-] (n21) -- (4,1);
\draw[line width=1pt,color=Alert] [-] (4,1) -- (n22) ;
\draw[line width=1pt,color=Alert] [-] (n22) -- (r4)
node[pos=0.5,below] {\scriptsize $ {x_{4}}  $};
\draw[line width=1pt,color=blue] [-] (t1) -- (n21) ;
\node[color=blue] at (2,2) {\small $ {\max\{x_{1},x_{2}\}}  $};
\node[color=Alert] at (1.8,0.5) {\small $ {\max\{x_{3},x_{4}\}}  $};
\draw[line width=1pt,color=Alert] [-] (t2) -- (n21) ;
\draw[line width=1pt,color=Alert] [-] (t2) -- (n21) ;
\draw[line width=1pt,color=blue] [-] (n13) -- (r1)
node[pos=0.5,above] {\scriptsize $ {x_{1}}  $};
\draw[line width=1pt,color=blue] [-] (n13) -- (r2) ;
\draw[line width=1pt,color=Alert] [-] (n13) -- (r3)
node[pos=0.5,below] {\scriptsize $ {x_{3}}  $};
\end{tikzpicture}
}
		\caption{\tb{Network with ``multicast/multirate'' architecture}}
		\label{fig:multicast} 	
	\end{figure}
\tb{In this context, an agent is a receiver (e.g., $ \{\text{Ri}\}_{\text{i}=1,2,3,4} $ in Fig.~\ref{fig:unicast}), who communicates with his respective transmitter via a fixed route consisting of links
	(e.g., R2 communicates with T1 via the route consisting of the T1-A, A-B, B-C, C-R2 links, whereas R3 communicates with T2 via the T2-A, A-B, B-C, C-R3 links). The scenario depicted in Fig.~\ref{fig:unicast} is such that receivers R1, R2 request the same data (e.g., the same movie), which is transmitted by transmitter T1. Similarly R3, R4 request the same data  which is transmitted by transmitter T2.
	The term ``unicast'' refers to the fact that the network establishes separate connections for each agent in each link
	(even when agents communicate the same content), and thus each agent loads each link with the amount of bandwidth it is allocated. The inequality constraints quantify the capacity constraints at each link (e.g., in link A-B, $x_1+x_2+x_3+x_4\leq c_{AB}$, where $c_{AB}$ is the capacity of link A-B and $x_i$'s are the allocated rates for the four agents sharing this link).
	Thus the unicast problem is a pure private goods problem with one good (i.e., rate) to be allocated  and multiple constraints on the same good (i.e., one capacity constraint per link).}

	There are however a number of interesting problems in economics and engineering that do not fit the above general model, as they involve agents forming groups and entail intergroup competition and intragroup sharing.
	We present this class of problems in the context of communication networks and in particular for the problem of rate allocation in a ``multicast/multirate'' architecture on the Internet, as depicted in Fig.~\ref{fig:multicast}.
	\tb{The model considered in this paper is such that for each receiver there is exactly one transmitter who communicates with it, whereas for each transmitter there can be multiple receivers who communicate with it.}
	\tb{Similarly to the unicast scenario, an agent is a receiver. Agents form multicast groups based on the content they communicate (e.g., \{R1, R2\} form a group and so do agents \{R3, R4\}). At the same time, agents within a group may request different bandwidth
	for the same content, where this differentiation is due to different quality of service requested by the users within a group, such as in high- vs standard-definition video.
	In the multicast protocol, in each link only a single connection is established for each multicast group carrying the corresponding content at the highest requested rate. This is motivated by the fact that any lower-quality content can always be derived from the higher-quality content and thus there is no need to further load the link other than with the highest rate (e.g., agents R1, R2, R3, R4 all communicate via link A-B, but only two distinct contents are transmitted on it, resulting in a capacity constraint of the form $\max\{x_{1},x_{2}\} + \max\{ x_{3} , x_{4} \} \leq c_{AB}$).}
	Such a unique architecture makes the multicast/multirate allocation problem qualitatively different from the unicast since there is sharing of bandwidth as well as competition for it.

	A second motivating scenario comes from the provision of data-security on server farms. Consider multiple server farms, each hosting data for several companies. Data security at each server farm is provided by simultaneous use of different security products out of a set $\mL$ of possible products. Due to different companies having different security needs, for any company $ i $ in server farm $ k $ achieving security level $ x_{ki} $ requires a profile $ \{ \alpha_{ki}^l x_{ki} \}_{l \in \mL} $ of quantities of different security products\footnote{The final security level can be thought of as being achieved by use of different security products via a Leontief-like production function~\cite[p.~49]{MWG}.}\tb{, where $ \alpha_{ki}^l $ are positive constants denoting the effectiveness of product $ l $ for company $ i $ on server farm $ k $}. However, since at each server farm the effectiveness of the security product is not additive, the quantity of any security product required at a server farm is dictated only by the maximum quantity/quality of that security product demanded amongst companies in that farm.
	Finally, allocation of security products among server farms is constrained by the limited quantity of each security product being available.

	Looking at the structure of both these problems, one sees both a private and a public goods aspect of market design in them.
	Referring to the multicast problem, due to the capacity constraints on the links of the network, allocation of rate for one content implies that such rate cannot \tb{be} allocated to agents sharing this link, who have requested a different content - this is the private good aspect resulting in intergroup competition.
	On the other hand for agents requesting the same content, since the allocation via the capacity constraint is dictated only by the highest rate from that group, others in the group can be allocated additional rate (up to the maximum requested in that group on that link) without having to affect anybody else's allocation - this is the public goods aspect resulting in intragroup sharing and the inevitable ``free-rider'' problem~\cite[sec.~11.C]{MWG}, albeit in a problem where consumption is still private.

\subsection{Contributions}

\tb{The goal of this paper is to design appropriate incentives, through allocation and taxes, such that when acting strategically i.e. at NE, the corresponding allocation to agents is efficient. This efficient allocation is the solution to a convex optimization problem: maximization of sum of agents' utilities subject to network multicast constraints. 
The taxes imposed by the mechanism refer to actual money paid by agents and not to virtual signals, as may be the case in works on distributed optimization.
Therefore our model assumes the presence of an authority capable of collecting taxes from the agents or disbursing subsidies to the agents (as dictated by the mechanism).}

\tb{For the problem of designing a mechanism that leads to a single-shot game, the well-accepted measure of complexity is the dimensionality of the message space of the proposed mechanism.} Similar to some of the works mentioned earlier, the mechanism presented here requires agents to communicate via announcing relatively ``small'' signals, which only consist of demands and prices.
However, contrary to some of the earlier works in communications, such as~\cite{rahuljain,hajekvcg,johari2009efficiency}, the work here ensures full implementation of social welfare maximizing allocation, so the designer can guarantee that only the most efficient outcome will be reached and no other. This is the first contribution of this work. \tb{Note that we avoid using direct mechanisms (see~\cite{nisan2007algorithmic,borgers2015book} for definition) since for this problem agents' private information is their utility function and it would be impractical to ask agents to quote entire functions. This also means that we do not use such terms  as ``incentive compatibility'' and ``truth-telling'' as they apply only to direct mechanisms.}

In addition to the stated goal of getting optimal allocations at all Nash equilibria (NE) and ensuring individual rationality (IR), the mechanism presented here has two auxiliary properties.
The first auxiliary property, and the second contribution of this work, is to achieve \textit{feasibility on and off equilibrium} via a ``radial projection'' allocation function. This property implies that the contract promises to allocate rates to agents in such a way that the link capacity constraints are never violated, \tb{contrary to the works in~\cite{stoenescu06,kakhbodcorrection,demosmulticorrection,demosshruti}}.
While, clearly, feasibility is satisfied by definition at NE, the proposed mechanism achieves it even off equilibrium. Off-equilibrium feasibility has much deeper implications in the context of networks as described below. The use of NE as a solution concept in the implementation literature can lead to problems in practice since calculation of NE  requires full information on part of the agents (which may not happen on an informationally and physically decentralized system like the Internet). Thus the justification for using NE (even for single shot games) is that for certain classes of learning dynamics, repeated play with the mechanism will ensure that the NE is ``learned'' eventually. Since information capacity constraints on a network are hard constraints and cannot be violated at any cost (in the short run), the above ``learning'' justification is applicable only when allocation is feasible throughout the learning process. Hence off-equilibrium feasibility becomes a necessity and the mechanism presented in this work is ``learning ready'' in that respect.

	Regarding the specific technique used to achieve off-equilibrium feasibility, radial projection refers to agents' demands being converted into allocations by scaling of the overall demand vector to the boundary of the feasible region defined by the capacity constraints. This allocation method subsumes as a special case the allocation function introduced in~\cite{basar,hajek} for communication on a single-link unicast network, and for stochastic control of networks in~\cite{kelly} (readers may be refer to~\cite{SiAn15b} for a full implementation mechanism that uses a similar allocation concept in the unicast problem).

	The second auxiliary property of the proposed mechanism and the final contribution of this work is to demonstrate how strong budget balance (SBB) at NE can be added to the non-budget balanced mechanism with the exchange of an extra signal (see Section~\ref{secmechown}). 
	This is in contrast to~\cite{demosmulticorrection}, where significant effort has been made to ensure budget balance off equilibrium for feasible allocations, while feasibility itself isn't ensured off equilibrium.
With the proposed modification, achieving SBB at NE becomes a relatively straightforward task. This is to be contrasted with the mechanisms proposed in~\cite{hajekvcg,johari2009efficiency} where SBB becomes a very difficult property to guarantee.

\subsection{Relevant literature}

The earliest unicast mechanisms were inspired by the seminal work~\cite{kelly}, which presents a mechanism with competitive equilibrium as the solution concept.
Mechanism design work for unicast network include~\cite{basarnash,basar,joharicongestion,hajek,rahuljain}, and specifically, Nash implementation for general unicast has been studied in~\cite{stoenescu06,kakhbodcorrection,SiAn15b}.
In all the above works, the network structure is assumed to be fixed; alternately~\cite{neely2009optimal} studies the relay problem where agents own links and can price the data going through them.

	Mechanism design has also been studied for power allocation problems in wireless networks. These are public goods problems, since the quality of service achieved by any agent depends not only on his signal strength but also that of neighboring agents due to interference from surrounding transmissions (see \cite{hong2012mechanism,demosshruti,alpcan_book2013}).

\tb{Readers may refer to~\cite{chen2002family,healy2012designing} for mechanisms designed to fully implement Walrasian and Lindahl allocations that are guaranteed to converge to the NE with a large class of learning dynamics. Both~\cite{chen2002family,healy2012designing} rely on the off-equilibrium properties - specifically that of the best response correspondence - to ensure a ``learning'' property in their mechanism. Authors in~\cite{bhattacharya2014} ensure that their mechanism induces a best response correspondence that has a partial truth-telling property: the allocation and price $ (x_i,p_i) $ quoted by agent $ i $, in her best response, are related to her utility $v_i(\cdot)$ as $ p_i = v_i^\prime(x_i) $ (even if $ x_i $ isn't equal to the optimal allocation $ x_i^\star $).
}

	In  multicast/multirate models without strategic agents, researchers have argued for different optimality criteria. One example is the \textit{max-min} fairness for multicast used in~\cite{rubenstein,srikantmulticast,sarkarfair,SaTa02}. On the other hand, the works in~\cite{sarkaropt,kar} have used integer and convex programming to get decentralized algorithms that maximize the sum of utilities. Stoenescu et al~\cite{stoenescu2007multirate} propose a realization algorithm which converges to optimal allocation.

	The remainder of this paper is structured as follows: in Section \ref{secCP}, the centralized problem that we wish to implement is stated and its solution is characterized. In Section \ref{secmech}, the mechanism is described and its properties are derived for the weak budget balance (WBB) case. This mechanism is modified in Section \ref{secmechown} to include strong budget balance at NE. Section \ref{secgen} discusses the relevant literature and some salient features of the mechanism.

	\section{Centralized Problem} \label{secCP}
	
	%

For concreteness, the exposition follows the prototypical example of rate
allocation in a multicast/multirate network architecture on the Internet.
	The system consists of a set $ \mN $ of Internet agents who communicate over a fixed multicast network. Each agent here is considered as a pair of source and destination users and the agents are divided into disjoint multicast groups based on the content that they communicate. The set of agents is described as $ \mN =\{(k,i) ~\vert~ k \in \mK, i \in \mG_k \} $, where the set of multicast groups is denoted by $ \mK = \{1,2,\ldots, K\} $ and within a group $ k \in \mK $, the set of agents is denoted by $ \mG_k = \{1,\ldots,G_k\} $. We use the notation $ ki $ instead of $(k,i)$ to denote a generic agent.

	We denote by $ x_{ki} \ge 0 $ the rate allocated to an agent $ ki $ and it refers to the data-rate that agent $ ki $ communicates with (with the data content being the same as that of any other agent from group $ k $).
	Thus the overall allocation for the system is a vector $ x = \left( x_{ki} \right)_{ki \in \mN} $ of rates allocated to all the agents.
	
	Agents' have private valuation and thus for any agent $ ki $, the valuation for allocation $ x_{ki} $ is $ v_{ki}(x_{ki}) $ where $ v_{ki}(\cdot):\  \mathbb{R}_{+}\rightarrow \mathbb{R} $.

	
	
	The multicast network consists of links through which agents' data is transmitted from source to destination. The route $ \mL_{ki} $ of agent $ ki $ is the set of links that agent $ ki $ uses for his communication, and $ \mL = \cup_{ki \in \mN} \mL_{ki} $ is the set of all available links. We denote by $ \mN^l $ the set of agents utilizing a link $ l \in \mL $ and it is defined as $ \mN^l = \{ki \in \mN ~\mid~ l \in \mL_{ki} \} $. Also we define $ \mG_k^l \subseteq \mG_k $ 
	to be the set of agents from group $ k $ who use link $ l $ and $ \mK^l $ to be the set of groups that have at least one agent that uses link $ l $ i.e. $ \mK^l = \{ k \in \mK ~\mid~ \mG_k^l \ne \emptyset \} $.
	
	Finally, the cardinality of all the sets defined above are denoted as follows $ N = \vert \mN \vert $, $ K = \vert \mK \vert $, $ G_k = \vert \mG_k \vert $, $ L = \vert \mL \vert $, $ L_{ki} = \vert \mL_{ki}  \vert $, $ G_k^l = \vert \mG_k^l \vert $ and $ N^l = \vert \mN^l \vert $.
	
	In addition to group-wise ordering of agents, we also introduce a combined group and link-wise ordering of agents\footnote{Given previous definitions, this reindexing is somewhat redundant; however, it will be used in the following to make the exposition clearer.}. For group $ k $ and link $ l $, any agent $ ki $, \tb{can be identified by the index} $g_k^l(i)$, where the mapping is defined as
	\begin{equation} \label{eqalternate}
	ki \mapsto g_k^l(i) \quad \text{where} \quad 1 \le g_k^l(i) \le G_k^l,  \quad \forall~l \in \mL_{ki}.
	\end{equation}
	The above mapping is defined so that it preserves the original ordering i.e. $ \forall $ $ i,j \in \mG_k^l $ and $ i > j $ we have $ g_k^l(i) > g_k^l(j) $. Note that with this requirement the mapping is uniquely defined. \tb{For instance, if group $ k $ has four agents $ \{ k1,k2,k3,k4 \} $ and only agents $ k1,k3,k4 $ use a particular link $ l $ then
	\begin{gather}
	k1 \mapsto g_k^l(1) = 1,~~ k3 \mapsto g_k^l(3) =2,~~ k4 \mapsto g_k^l(4) = 3.
	\end{gather}
}
	


	The network administrator is interested in maximizing the social welfare under the link capacity constraints. This centralized problem is
	\begin{alignat}{2} \label{CP0}
	&\max_x \sum_{k \in \mK} \sum_{i \in \mG_k} v_{ki}(x_{ki}) \tag{CP$_0$}\\
	\text{s.t.} &\quad x_{ki} \ge 0 ~~\forall~ ki \in \mN &&  \tag{I$_1$}\\
	\text{and}  &\quad \sum_{k \in \mK^l} {} \max_{j \in \mG_k^l}\{\alpha^l_{kj} x_{kj}\} \le c^l ~~\forall~ l \in \mL. &&  \tag{I$_2$}
	\end{alignat}
	
	Constraints $ I_2 $ are the inequality constraints on allocation and represent the capacity constraint for every link $ l \in \mL$, in the network.
	Here $ \alpha^l_{kj} $ are constants that represent the quality of service requirement of agent $ kj $ combined with the specific architecture on link $ l $. So in order for agent $ kj $ to receive an actual data-rate of $ x_{kj} $ the bandwidth spent on link $ l $ is $ \alpha_{kj}^l x_{kj} $.
	As an example, $\alpha^l_{kj}=\frac{1}{R_{kj}(1-\epsilon^l_{kj})}$ for all links $l\in\mL_{kj}$, where $\epsilon^l_{kj}$ represents the packet error probability for link $l$ for a packet encoded with channel coding rate $R_{kj}$.
	Observe that due to the   multicast/multirate architecture, only the maximum rate of each group $k$ at each link $l$ enters the capacity constraints.

	\subsection{Assumptions} The analysis will be done under the following assumptions.
	
	\begin{enumerate}
		\item[(A1)] For all agents, $ v_{ki}(\cdot) \in \mathcal{V}_{ki} $, where the sets $\mathcal{V}_{ki}$ are arbitrary subsets of $\mathcal{V}_0$, the set of all strictly increasing, strictly concave, twice differentiable functions $\mathbb{R}_{+} \to \mathbb{R}$ with continuous second derivative.
		
		\item[(A2)] $ v_{ki}^{\prime}(0) $ is finite $ \forall $ $ ki \in \mN $. This also implies that $ v_{ki}^{\prime}(x) $ is finite and bounded $ \forall $ $ ki $ and $ \forall $ $ x \in \mathbb{R}_+ $ since $ v_{ki} $'s are concave.
		
		\item[(A3)] Every link has at least two groups that use it i.e. $ K^l \ge 2 $ $ ~\forall $ $ l \in \mL $.
		
		\item[(A4)] The optimal solution of the centralized problem is such that for every link there are at least 2 groups such that each has at least one non-zero component, i.e. if $ S^l(x) \coloneqq \{ k \in \mK^l ~\mid~  \exists ~i \in \mG_k^l ~\text{ s.t. } x_{ki} > 0  \} $ then the assumption says $ \vert S^l(x^{\star}) \vert \ge 2 $ $ ~\forall $ $ l \in \mL $ (where $ x^{\star} $ is the optimal solution of \eqref{CP0}).
	\end{enumerate}
	In addition, the coefficients are all strictly positive, i.e. $ \alpha_{ki}^l > 0 $ $ ~\forall $ $ l \in \mL_{ki} $, $ \forall $ $ ki \in \mN $. 
	Also, for well-posedness of the problem we take $ c^l > 0 $ $ ~\forall $ $ l \in \mL$.
	
	Assumption (A1) is made in order for the centralized problem to have a unique solution at the boundary of the feasible region and for this solution to be precisely characterized by the KKT conditions.
	(A2) is a mild technical assumption that is required in the proof of Lemma~\ref{lemexis}.
	Assumption (A3) is  made in order to avoid situations where there is a link/constraint involving only one multicast group.
	Such a case requires special handling in the design of the mechanism (since in such a case there is no contention at that link), and destructs from the basic idea
	that we want to communicate.
	Finally (A4) is related to (A3) and is made in order to simplify the exposition of the proposed mechanism,
	without having to define corner cases that are of minor importance. \tb{One can state a number of mild sufficient conditions on the original optimization problem~\eqref{CP0} such that (A4) is satisfied. For instance if the optimal solution $ x^\star = \big( x_{ki}^\star \big)_{ki \in \mN} $ is such that $ x_{ki}^\star > 0 $ for every agent $ ki $ then assumption (A4) is satisfied. Furthermore, this can be ensured by considering utility functions such that $ v_{ki}^\prime(0) $ is sufficiently large.}

	\subsection{Necessary and Sufficient Optimality conditions}
	\tb{Under the stated assumptions,~\eqref{CP0} is a convex optimization problem.} Following are the necessary KKT conditions for optimality, which under the stated assumptions are also sufficient (\tb{due to the strict concavity of $ v_{ki}(\cdot) $}).
	For this the centralized problem is first rewritten by restating the capacity constraints in a different form
	\begin{alignat}{3} \label{CP}
	&\max_{x,m} \sum_{k \in \mK} \sum_{i \in \mG_k} v_{ki}(x_{ki}) \tag{CP} \\
	\text{s.t.} &\quad x_{ki} \ge 0 ~~\forall~ ki \in \mN &&  \tag{C$_1$} \\
	\text{and}  &\quad \sum_{k \in \mK^l} m_k^l \le c^l  ~~\forall~ l \in \mL &&\tag{C$_2$}\\
	\text{and}	&\quad \alpha_{ki}^l x_{ki} \le m_k^l ~~\forall~ i \in \mG_k^l,~k \in \mK^l,~ l \in \mL. && \tag{C$_3$}
	\end{alignat}
	The capacity constraints have been rewritten with the introduction of new variables. The virtual variables $ \tb{m_k^l \in \mathbb{R}_+} $ represent the weighted maximum requirement of group $ k $ on link $ l $. It's easy to see that the solution of~\eqref{CP} is the same as the solution of the original~\eqref{CP0} as far as optimal $ x $ is concerned. 
We introduce $ \lambda,\mu,\nu $ as the dual variables corresponding to constraints $C_2$, $C_3$ and $C_1$, respectively. The KKT conditions are stated without explicitly referring to $ \nu_{ki} $'s and just using the fact that $ \nu_{ki}^{\star} \ge 0 $ and $ \nu_{ki}^{\star} x_{ki}^{\star} = 0 $ $ ~\forall~ki \in \mN $.
	Note that with the assumptions above, the KKT conditions below will give rise to a unique $ x^{\star} $
	as the optimizer for \eqref{CP}.
	
	\underline{KKT conditions}: \begin{enumerate}
		\item[a)] Primal Feasibility:
		\begin{subequations}
		\begin{gather}
        x_{ki}^{\star} \ge 0 ~~\forall~ ki \in \mN\\
		\sum_{k \in \mK^l} {m_k^l}^{\star} \le c^l ~~\forall~ l \in \mL;  \\
		\alpha_{ki}^l x_{ki}^{\star} \le {m_k^l}^{\star} ~~ \forall~ i \in \mG_k^l, ~ k \in \mK^l, ~ l \in \mL.
		\end{gather}
		\end{subequations}
		\item[b)] Dual Feasibility:  $ \lambda^{\star}_l \ge 0 $ $ ~\forall $ $ l \in \mL $; $~  \mu_{ki}^l \ge 0 $ $ ~\forall $ $ ki \in \mN^l $, $ l \in \mL $.
		
		\item[c)] Complimentary Slackness:
		\begin{subequations}
		\begin{align}
		\lambda^{\star}_l \left( \sum_{k \in \mK^l} {m_k^l}^{\star} - c^l \right) &= 0 ~~\forall~ l \in \mL, \\
		\label{EQCS2}
		{\mu_{ki}^l}^{\star} \left(\alpha_{ki}^l x_{ki}^{\star} - {m_k^l}^{\star}  \right) &= 0 ~~\forall~ i \in \mG_k^l,~ k \in \mK^l,~ l \in \mL.
		\end{align}
	\end{subequations}
		
		\item[d)] Stationarity:
	\begin{subequations}	
		\begin{gather} \label{EQSS1}
		v_{ki}^{\prime}(x_{ki}^{\star}) = \sum_{l \in \mL_{ki}} {\mu_{ki}^l}^{\star} \alpha^l_{ki} ~ \forall~ ki \in \mN \quad \text{if} \quad x_{ki}^{\star} > 0 \\
		v_{ki}^{\prime}(x_{ki}^{\star}) \le \sum_{l \in \mL_{ki}} {\mu_{ki}^l}^{\star} \alpha^l_{ki} ~\forall~ ki \in \mN \quad \text{if} \quad x_{ki}^{\star} = 0
		\end{gather}
	\end{subequations}
		and
		\begin{equation} \label{EQSS2}
		\lambda_l^{\star} = \sum_{i \in \mG_k^l} {\mu_{ki}^l}^{\star} \quad \forall~~ k \in \mK^l, ~~ l \in \mL.
		\end{equation}
	\end{enumerate}
	Looking at \eqref{EQCS2}, $ {\mu_{ki}^l}^{\star} $ will be non-zero only if $ \alpha_{ki}^l x_{ki}^{\star} = {m_k^l}^{\star} $, so each $ {\mu_{ki}^l}^{\star} $ can be interpreted as the ``price'' paid only by those agents who receive maximum weighted allocation in group $k$ at a given link $l$. Consequently, from \eqref{EQSS2}, $ \lambda_l^{\star} $ is the sum of $ {\mu_{ki}^l}^{\star} $ over those agents in group $k$ that get maximum allocation within the group and it is the same for all groups. $ \lambda_l^{\star} $ can then be interpreted as the common total price per unit of rate that each group is subject to at link $ l $.

	\section{A mechanism with weak budget balance} \label{secmech}

In this section, we consider the case of weak budget balance, i.e., the taxes $t_{ki}$ of all agents at equilibrium satisfy $\sum_{ki\in \mN} t_{ki}\geq 0$. \tb{Note that the taxes used in the model refer to actual money paid by agents and not virtual signals, as may be the case in other works on distributed optimization.}
	
	\subsection{Mechanism } \label{subsecmech}

The designer designs and announces the message space $ \mathcal{S}_{ki} $ for each agent $ ki \in \mN $. The agents pick their message simultaneously and broadcast it. Based on the message profile received $ s = \left(s_{ki}\right)_{ki \in \mN} $, the designer allocates rate $ x = \left(x_{ki}\right)_{ki \in \mN} $ as $ x_{ki} = h_{ki}^x(s) $. Similarly, the designer levies a tax (or subsidy) on each agent $ t = \left(t_{ki}\right)_{ki \in \mN} $ as $ t_{ki} = h_{ki}^t(s) $. The proposed mechanism is described below by defining the sets $ \mathcal{S}_{ki} $ and functions $ \left(h_{ki}^x(\cdot),h_{ki}^t(\cdot)\right) $, for each $ ki \in \mN $.

This gives rise to a one-shot game
\begin{equation}
\mathfrak{G} = \Big( \mN,\big(\mathcal{S}_{ki}\big)_{ki \in \mN}, \big(\widehat{u}_{ki}(\cdot)\big)_{ki \in \mN} \Big),
\end{equation}
played by all the agents in $ \mN $, where action sets are
$ (\mathcal{S}_{ki})_{{ki} \in \mN} $ and utilities are given by
\begin{equation}
\widehat{u}_{ki}(s) = v_{ki}(x_{ki}) - t_{ki} = v_{ki}(h^{x}_{ki}(s)) - h^{t}_{ki}(s) ~~\forall~ki \in \mN.
\end{equation}
The resource allocation problem~\eqref{CP} would be \textrm{\textit{fully implemented in NE}}, if the outcomes (all possible NE) of this game  produce allocation $ x^{\star} $ and all agents in $ \mN $ are better-off participating in the mechanism than opting out (getting $ 0 $ allocation and taxes). 

\noindent
\textbf{Message Space. }
The designer asks each agent to report a message $ s_{ki}= (y_{ki}, Q_{ki}) $ where $ Q_{ki} = \left({^1Q_{ki}}, {^2Q_{ki}} \right) $ with $ {^1Q_{ki}} = \left( ^1q_{ki}^l \right)_{l \in \mL_{ki}} $ and $ {^2Q_{ki}} = \left( ^2q_{ki}^l \right)_{l \in \mL_{ki}} $; also for convenience denote $ q_{ki}^l = \left({^1q_{ki}^l}, {^2q_{ki}^l}\right) $. The message $ s_{ki} $ includes a proxy for agent $ ki $'s demand for the rate, $y_{ki} \in \mathbb{R}_+$, and two ``prices'' $ {^1q_{ki}^l} \in \mathbb{R}_+ $, $ {^2q_{ki}^l} \in \mathbb{R}_+ $ for each link $ l \in \mL_{ki} $ that agent $ki$ is involved in.
For each agent $ki$ and each link $l$, the first price $ ^1q_{ki}^l $ relates to the constraint $ \alpha_{ki}^l x_{ki} \le m_k^l $, while the second price $ {^2q_{ki}^l} $ relates to the constraint $ \alpha_{kj}^l x_{kj} \le m_k^l $, where $ kj $ is the agent \tb{from group $ k $ on link $ l $ which can be identified by the index $ g_k^l(i) + 1 $  (see the group and link-wise indexing notation defined in~\eqref{eqalternate})}.
As it turns out, our methodology makes quoting two prices necessary, as will be shown later. \tb{An intuitive explanation for the need of two prices is provided in the Discussion paragraphs after the tax definition in~\eqref{eq:taxes1} and~\eqref{EQtaxsecond}.}
	From the above definitions, the message space for agent $ ki $ is $ \mathcal{S}_{ki} = \mathbb{R}_{+} \times \mathbb{R}_{+}^{2L_{ki}} $.
For received messages $ s = (s_{ki})_{ki\in\mN} = (\novec{y}, \matrix{Q}) = ((y_{ki})_{ki\in\mN}, (\novec{Q}_{ki})_{ki\in\mN}) $ the contract $ h_{ki}(s) = (h^{x}_{ki}(s), h^{t}_{ki}(s)) $ has an allocation component $h^{x}_{ki}(\cdot)$ and a tax component\footnote{All utilities in this work are assumed to be quasi-linear.}  $h^{t}_{ki}(\cdot)$ and is defined for each $ ki \in \mN $ as follows.

\noindent
\textbf{Allocation function.}
If the received demand vector is $ y = (y_{11}, \ldots, y_{K G_K}) = 0 $ then the allocation is $ x = (x_{11}, \ldots, x_{K G_K}) = 0 $ (also set $ m = 0 $).
Otherwise it is evaluated by radially projecting the demand vector $y$ on the boundary of the feasibility region as depicted in Fig.~\ref{fig:multicast_allocation} for a two-group single-link case.
	\begin{figure}[htbp]
		\centering
		\tdplotsetmaincoords{75}{110}
\begin{tikzpicture}[tdplot_main_coords,scale=4]
\draw[thick,->] (0,0,0) -- (1.25,0,0) node[anchor=north east]{$x_{11}$};
\draw[thick,->] (0,0,0) -- (0,1.25,0) node[anchor=north west]{$x_{12}$};
\draw[thick,->] (0,0,0) -- (0,0,1.25) node[anchor=south]{$x_{21}$};
\draw[dashed,red] (0,0,0) -- (1,1,0) ;
\draw[red,thick] (0,0,1) -- (1,0,0) ;
\draw[red,thick] (0,0,1) -- (0,1,0) ;
\draw[red,thick] (0,0,1) -- (1,1,0) ;
\draw[red,thick] (1,0,0) -- (1,1,0) -- (0,1,0) ;
\draw[blue,thick,dashed] (0,0,0) -- (0.47619,0.61904,0.38095) node[anchor=north west]{\color{black}{$ x $}} ;
\draw[Green,thick,dashed] (0.47619,0.61904,0.38095) -- (1,1.3,0.8) node[anchor=south]{\color{black}{$ y $}} ;
\draw[blue,thick] (0.47619,0.61904,0.38095) circle (0.008cm);
\draw[Green,thick] (1,1.3,0.8) circle (0.008cm);
\draw[dotted,thick]  (0.47619,0.61904,0.38095) -- (0,0.61904,0.38095) ;
\draw[dotted,thick]  (0.47619,0.61904,0.38095) -- (0.47619,0,0.38095) ;
\draw[dotted,thick]  (0.47619,0.61904,0.38095) -- (0.47619,0.61904,0) ;
\end{tikzpicture}
		\caption{Radial projection allocation for a single-link multicast/multirate network with 2 multicast groups $ \mK = \{1,2\} $ with a total of 3 agents i.e. $ \mN = \{11,12,21\} $. Capacity constraint is $ \max \{ x_{11},x_{12} \} + x_{21} \le 1 $. For demand $ y $ allocation $ x $ is just the projection of $ y $ onto the feasible set boundary (first sub-case of~\eqref{EQA1} applies here).}
		\label{fig:multicast_allocation} 	
	\end{figure}
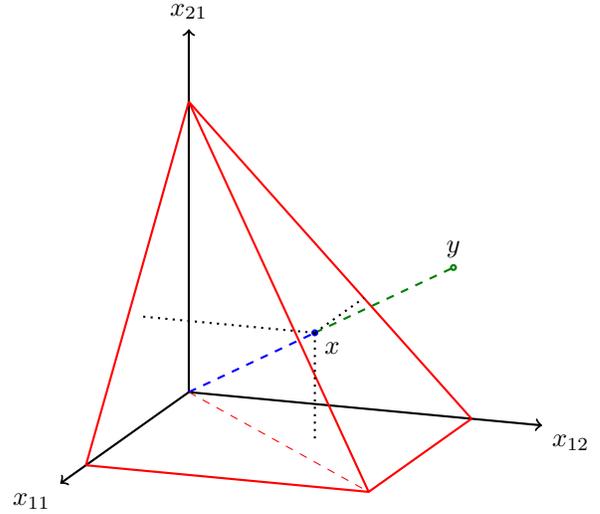
More formally, the allocation function $ h^{x}_{ki}(s) $ creates allocation by first creating proxies \tb{$ n_k^l \in \mathbb{R}_+ $} for the weighted maximum at each link for each group as follows
\begin{subequations}
	\begin{gather}
	n_k^l \triangleq \max_{i \in \mG_k^l} \: \{\alpha_{ki}^l  y_{ki} \} ~~\forall~k \in \mK^l,~ l \in \mL .
	\end{gather}
Then intermediate variables $ m_k^l $'s are created by dilating/shrinking $ n_k^l $'s on to one of the hyperplanes defined by the second set of constraints in $ C_2 $, specifically, that hyperplane for which the corresponding $ m_k^l $'s are the closest to origin (this could also be at the intersection of multiple hyperplanes). This is done through the introduction of a scaling factor $r$ as follows
\begin{gather}
	\label{EQx0}
	r = \min_{l \in \mL} \: r^l
	\end{gather}
	\begin{align} \label{EQA1}
	r^l &= \left\{
	\begin{array}{ll}
	\frac{c^l}{\sum_{k \in \mK^l}  n_k^l }, & \quad \text{if }~ \vert S^l(y) \vert \ge 2 \\[1ex]
	\frac{c^l}{\sum_{k \in \mK^l}  n_k^l } - f^l(n_k^l), & \quad \text{if }~  S^l(y) = \{k\}  \\
	+\infty & \quad \text{if}~ \vert S^l(y) \vert = 0
	\end{array}
	\right.
	\end{align}
	\begin{equation}
	f^l(n_k^l) =
	\frac{c^l}{n_k^l (n_k^l + 1)}.
	\end{equation}
\end{subequations}
Finally allocation $ x_{ki} $ is calculated by dilating/shrinking $ y_{ki} $ by the same factor.
\begin{subequations}
	\begin{align} \label{EQx}
	h^{x}_{ki}(s) &= x_{ki} = r \,y_{ki} ~~\forall~ki \in \mN,
	\\
	m_k^l &\triangleq r\, n_k^l, ~~\forall~k \in \mK^l,~ l \in \mL.
	\end{align}
\end{subequations}
In other words, the contract dilates/shrinks $ n_k^l $ to the boundary of the feasible region defined by the capacity constraints and then allocations within a group are made proportionally. Since all the $ \alpha^l_{kj} $'s are positive, this means that all constraints in $ C_2 $ are satisfied for the allocation automatically, as shown later.
In the above description of the allocation function, the separate definition for $ r^l $ when $ \vert S^l(y) \vert < 2 $ is to ensure (as it will be shown later) that there are no equilibria with allocation $x^{\star}$, where $ \vert S^l(x^{\star}) \vert < 2 $.

\noindent
\textbf{Tax function.}
For the taxes, we first define total prices, $ w_k^l, \bar{w}_{-k}^l $ for any link $ l $ and group $ k \in \mK^l $ as
\begin{align} \label{EQpm}
w_k^l \triangleq \sum_{i \in \mG_k^l} {} {^1q_{ki}^l}, \qquad \bar{w}_{-k}^l \triangleq  \frac{1}{ K^l- 1 } \sum_{k^{\prime} \in \mK^l \backslash \{k\} } w_{k^{\prime}}^l
\end{align}
where $ \bar{w}_{-k}^l $ is well-defined due to assumption (A3).
The tax is defined as the sum of taxes for each constraint in an agent's route
\begin{gather}
\label{EQt}
h^{t}_{ki}(s) = t_{ki} = \sum_{l \in \mL_{ki}} t_{ki}^l~~\forall~ki \in \mN,
\end{gather}
and each component $t_{ki}^l$ is defined as follows.
If $ G_k^l \ge 2 $ then consider agents $ kj $ and $ ke $ \tb{from group $ k $ and on link $ l $ who can be identified by the index $ g_k^l(i) - 1 $ and $ g_k^l(i) + 1 $ (mod $ G_k^l $), respectively}. Then,
\begin{multline}
\label{eq:taxes1}
t_{ki}^l = x_{ki} \alpha^l_{ki}  {^2q_{kj}^l} +  ({^2q_{ki}^l} -  {^1q_{ke}^l})^2 + (w_k^l - \bar{w}_{-k}^l)^2
\\
+ \eta \, {^2q_{kj}^l} ({^1q_{ki}^l} -  {^2q_{kj}^l})(m_k^l - \alpha_{ki}^l x_{ki})
\\
{} + \xi \, \bar{w}_{-k}^l (w_k^l - \bar{w}_{-k}^l) (c^l- \sum_{k^{\prime} \in \mK^l} m_{k^{\prime}}^l),
\end{multline}
where $ \eta,\xi $ are sufficiently small positive constants (whose selection is outlined in the proof of Lemma~\ref{lemexis}). If $ G_k^l = 1 $ then
\begin{multline} \label{EQtaxsecond}
t_{ki}^l = x_{ki} \alpha^l_{ki} \bar{w}_{-k}^l + (w_k^l - \bar{w}_{-k}^l)^2
\\
+ \eta \, \bar{w}_{-k}^l ({^1q_{ki}^l} -   \bar{w}_{-k}^l)(m_k^l - \alpha_{ki}^l x_{ki})
\\
{} + {}  \xi \, \bar{w}_{-k}^l (w_k^l - \bar{w}_{-k}^l) (c^l- \sum_{k^{\prime} \in \mK^l} m_{k^{\prime}}^l).
\end{multline}


\noindent
\textbf{Discussion.}
\tb{
The need for two prices is mainly a technical necessity (an unavoidable corner case whenever an inequality constraint involves only one user) and can be explained as follows. When there is a constraint involving at least two agents, then it suffices to ask each agent to quote a single price for it. The price quoted by agent $kj$ would be used is used in two places in~\eqref{eq:taxes1}: it would determine the price paid by agent $ki$
in his first tax term (this way agent $ki$ does not control both his price and (indirectly) his allocation)
and it will also be used in the fourth term of agent's $kj$ taxes to make sure the complementary slackness condition is satisfied. Similarly for the price quoted by user $ki$.
However, for the problem at hand, constraints (C$_3$) ($ \alpha_{ki}^l x_{ki} \le m_k^l $) involve only one agent, $ki$, per constraint. This changes things because this quoted price cannot be used both in his first and fourth tax terms! This necessitates that users quote two prices as follows.
We require agent $ ki $, to quote a price $ ^1q_{ki}^l $ that will be used in his fourth tax term to  make sure the complementary slackness condition is satisfied.
Clearly we cannot use the same quoted price for his first tax term (this would violate the condition that the same agent cannot control both his price and quantity).
So we ask agent $kj$ to quote a price $ ^2q_{kj}^l $ that is paid by agent $ ki $ (first term in $ t_{ki}^l $).
These two prices, $ ^2q_{kj}^l, {^1q_{ki}^l} $ (the second price quoted by $kj$ and the first price quoted by $ki$) are thus used as a proxy for the same quantity: the optimal Lagrange multiplier $ {\mu_{ki}^l}^\star $ corresponding to  the constraint $ \alpha_{ki}^l x_{ki} \le m_k^l $.
}

\tb{
It should now be clear what the role of the other tax terms is. The second and third terms $ ({^2q_{ki}^l} -  {^1q_{ke}^l})^2, (w_k^l - \bar{w}_{-k}^l)^2 $ are introduced to incentivize agents to quote same prices.
For instance, the term $ ({^2q_{ki}^l} -  {^1q_{ke}^l})^2 $ drives agent $ki$ to quote $ {^2q_{ki}^l} $ such that it matches $ {^1q_{ke}^l} $ - both are a proxy for the optimal Lagrange multiplier $ {\mu_{ke}^l}^\star $, as explained above (for the $kj, ki$ pair of agents).
Similarly, the purpose of the term $ (w_k^l - \bar{w}_{-k}^l)^2 $ is to tax every agent from group $ k $ additionally (in a smooth manner) if the total group price $ w_k^l = \sum_{i \in \mG_k^l} {^1q_{ki}^l} $ doesn't match with the average of group price $ \bar{w}_{-k}^l $ for other groups. This will drive agents to quote $ {^1q_{ki}^l} $'s such that the two match and if the group price for all groups match with the average of others, then indeed all group prices are the same (as required by~\eqref{EQSS2} in KKT).
To satisfy the complimentary slackness conditions at NE we introduce the fourth and fifth terms which charge agent $ ki $ higher taxes if they quote non-zero prices for inactive constraints. This necessarily requires using an agent's own quoted price (i.e. price $ ^1q_{ki}^l $ for agent $ ki $).
}

\tb{
Finally, below we show that all tax terms except the first tax term are zero at NE and $ ^1q_{ki}^l = {^2q_{kj}^l} = {\mu_{ki}^l}^\star $, where $ {\mu_{ki}^l}^\star $ is the optimal Lagrange multiplier (see~\eqref{EQCS2},~\eqref{EQSS1}). Consequently,  the total tax paid by agent $ ki $ at NE is $ t_{ki} = x_{ki}^\star \sum_{l \in \mL_{ki}} \alpha_{ki}^l {\mu_{ki}^l}^\star = x_{ki}^\star v_{ki}^\prime(x_{ki}^\star) $ (see~\eqref{EQSS1}). The price paid by agent $ ki $ at NE is the true marginal valuation $ v_{ki}^\prime(x_{ki}^\star) $ of agent $ ki $ and thus the mechanism ensures fairness of taxes paid by agents in the sense that it they are exactly what expected in a free market.
}

Readers may refer to~\cite{hurwicz1979outcome,chen2002family} for mechanisms with small message spaces which use similar technique for taxation in Nash implementing mechanisms for Lindahl correspondence (public goods).
	
There are two levels of interactions that this mechanism is dealing with, one among groups for allocation of maximum rate on each link and a second within each group. Agents contest for allocation that makes full use of the fact that from within a group only maximum at each link will give rise to a positive price on that link. At any link $ l $ and group $ k $, total price $ w_k^l $ is the sum of prices quoted by all the agents in the group at link $ l $. The quantity $ \bar{w}_{-k}^l $ is calculated by averaging the total prices for link $ l $ over all other groups than $ k $.
Quoting of prices and demand is used as a way of eliciting $ v_{ki}^{\prime}(x_{ki}) $ by comparing it appropriately with prices. In this vein, we do not wish to influence $ \bar{w}_{-k}^l $ with prices quoted by groups whose agents aren't using the link at all, since the price then essentially doesn't contain any information.

	\subsection{Results}

The basic result of this section is summarized in the following theorem.
	\begin{theorem}[Full Implementation] \label{thmain}
		For game $ \mathfrak{G} $, there is a unique allocation, $ x $, corresponding to all NE. Moreover, $ x = x^{\star} $, the maximizer of \eqref{CP}. In addition, individual rationality is satisfied for all agents and so is weak budget balance.
	\end{theorem}
	
	The theorem is proved by a sequence of results, in which all candidate NE of $ \mathfrak{G} $ are characterized by necessary conditions until only one family of NE candidates is left. Subsequently,  the existence of NE in pure strategies for $ \mathfrak{G} $ is shown, and that all NE result in allocation $ x = x^{\star} $. 
	Finally, individual rationality and WBB will be checked.
	
	
	\begin{lemma}[Primal Feasibility] \label{lemPF}
		For any action profile $ s = (\novec{y}, \matrix{Q}) $ of game $ \mathfrak{G} $, constraints $ C_1 $, $ C_2 $ and $C_3$ of \eqref{CP} are satisfied at the corresponding allocation.
	\end{lemma}
	\begin{proof}
		Please see Appendix~\ref{app1}.
	\end{proof}

	Feasibility of allocation for action profiles is a direct consequence of the radial projection allocation function. Next, it is shown that all groups, using a link, quote the same total price $ w_k^l $ for that link at any equilibrium. This is brought about by the 3$^{\text{rd}} $ tax term $ \sum_l (w_k^l - \bar{w}_{-k}^l)^2 $. This is a way of threatening agents with higher taxes just for quoting a different price than average,  at each link.

	\begin{lemma} \label{lemcmp}
		At any NE $ s = (\novec{y}, \matrix{Q}) $ of $ \mathfrak{G} $, for any link $ l \in \mL $ we have
		\begin{equation}
		w_k^l = w^l \quad \forall~ k \in \mK^l.
		\end{equation}
		Also, for any group $ k $ and link $ l $ such that $ G_k^l \ge 2 $ if we take any agents $ i, e \in \mG_k^l $ \tb{where agent $ ke $ can be identified by the index $ g_k^l(i) + 1 $} then at equilibrium we will have $ {^2q_{ki}^l} = {^1q_{ke}^l} $ (which will be denoted as $ p_{ke}^l $).
	\end{lemma}
\begin{proof}
	Please see Appendix~\ref{app2}.
\end{proof}
	
	
	With Lemma \ref{lemcmp}, equilibria can now be referred to in terms of the common total price vector $P$ rather than the two different total price vectors $^1Q$, and $^2Q$. In particular, any NE $ s = (y, Q) $ with $ Q = \left( {^1Q} , {^2Q} \right) $ can be characterized as $ s = (y, P) $ with $ P = (P_{ki})_{ki \in \mN} $ and $ P_{ki} = \left( p_{ki}^l \right)_{l \in \mL_{ki}} $.

	Later it will become clear how $ p_{ki}^l $ and $ w^l $ take the place of dual variables $ \mu_{ki}^l $ and $ \lambda_l $, respectively, when we compare equilibrium conditions with KKT conditions, hence we identify the following condition as dual feasibility.

	\begin{lemma}[Dual Feasibility] \label{lemDF}
		$ p_{ki}^l \ge 0$, $ w^l \ge 0 $ $ ~~\forall $ $ i \in \mG_k^l $, $ \forall $ $ k \in \mK^l $ and $ ~\forall $ $ l \in \mL $.
	\end{lemma}
	\begin{proof}
		This is by design, since agents are only allowed to quote non-negative prices and that $ w^l $ is the sum of such prices.
	\end{proof}

	Following is the property that solidifies the notion of prices as dual variables, since the claim here is that inactive constraints do not contribute to payment at equilibrium. This notion is very similar to the centralized problem, where if we know certain constraints to be inactive at the optimum then the same problem without these constraints would be equivalent to the original. The 4\textsuperscript{th} and 5\textsuperscript{th} terms in the tax function facilitate this by charging extra taxes for inactive constraints where the agent is quoting higher prices than the average of remaining ones, thereby driving prices down.

	\begin{lemma}[Complimentary Slackness] \label{lemCompSlac}
		At any NE $ s = (\novec{y}, P)$ of game $ \mathfrak{G} $ with corresponding allocation $ x $, for any agent $ i \in \mG_k^l $, group $ k \in \mK^l $ and link $ l \in \mL $ we have 
		\begin{equation}
		w^l \left ( \sum_{k \in \mK^l} m_k^l - c^l \right) = 0,  \qquad  p_{ki}^l \left( \alpha_{ki}^l x_{ki} - m_k^l \right) = 0.
		\end{equation}
	\end{lemma}
	\begin{proof}
		Please see Appendix~\ref{app3}.
	\end{proof}

	\begin{lemma}[Stationarity] \label{lemStat}
		At any NE $ {s} = ({\novec{y}}, \novec{P}) $ of game $ \mathfrak{G} $, and corresponding allocation $ {x} $, we have
		\begin{subequations}
		\begin{gather} \label{EQKKTst1}
		v_{ki}^{\prime}({x}_{ki}) = \sum_{l \in \mL_{ki}} p_{ki}^l \alpha^l_{ki} ~~ \forall~ ki \in \mN \quad \text{if} \quad {x}_{ki} > 0
		\\ \label{EQKKTst2}
		v_{ki}^{\prime}({x}_{ki}) \le \sum_{l \in \mL_{ki}} p_{ki}^l \alpha^l_{ki} ~~ \forall~ ki \in \mN \quad \text{if} \quad {x}_{ki} = 0
		\end{gather}
		\end{subequations}
		and
		\begin{equation} \label{EQS2}
		w^l = \sum_{i \in \mG_k^l} p_{ki}^l ~~\forall~ k \in \mK^l ,~ l \in \mL.
		\end{equation}
	\end{lemma}
	\begin{proof}
		Please see Appendix~\ref{app4}.
	\end{proof}
	
	Collecting the results of the above lemmas, we can conclude that every NE profile satisfies the KKT conditions of the \eqref{CP}.
	This means there are now necessary conditions on the NE up to the point of having unique allocation (since KKT for~\eqref{CP} is satisfied by a unique optimal $ x $). In the next Lemma we verify the existence of the equilibria that we have claimed.

\begin{lemma}[Existence] \label{lemexis}
For the game $ \mathfrak{G} $, there exists an equilibrium.
\end{lemma}
	\begin{proof}
		Please see Appendix~\ref{app5}.
	\end{proof}

	Several comments are in order regarding the selection of the radial projection allocation function and in particular~\eqref{EQA1}.
	If one uses ``pure'' radial projection i.e. same expression for $ r^l $ for $ \vert S^l(y) \vert \ge 2 $ and $ \le 1 $, then irrespective of the optimal solution of \eqref{CP}, for game $ \mathfrak{G} $ the ``stationarity'' property will not be satisfied at equilibria with $ \vert S^l(y) \vert \le 1$. Thus the mechanism will result in additional extraneous equilibria.
	This is the reason why we tweak the expression for $ r^l $ when $ \vert S^l(y) \vert \le 1 $, so as to eliminate these extraneous equilibria - irrespective of the solution of \eqref{CP}.
	With this tweak in the expression for $ r^l $, all KKT conditions become necessary for all
	equilibria regardless of the value of  $\vert S^l(y) \vert $.
	This however creates a problem in the proof of existence of equilibria.
	In particular, if $ x^{\star} $ was such that it had links where $ \vert S^l(x^{\star}) \vert = 1 $ then with modified radial projection allocation this would require $ y $ at NE such that $ \vert S^l(y) \vert = 1 $.
	In this case the $ r^l $ used would be lower than what the radial projection requires (see second sub-case in \eqref{EQA1}) and we actually would have the problem of possibly not having any $ y $ that creates $ x^{\star} $ as allocation.
	Hence we have used (A4) to eliminate this case.

	\begin{lemma}[Individual Rationality and WBB] \label{lem_ir}
		At any NE $ s = (\novec{y}, \novec{P}) $ of $ \mathfrak{G} $, with corresponding allocation $ x $ and taxes $ t $, we have
		\begin{align} \label{EQIRB}
		u_{ki}(x,t) &\ge u_{ki}(0,0) ~~\forall~ ki \in \mN \\
		\text{and} \qquad \sum_{ki \in \mN} t_{ki} &\ge 0.  
		\end{align}
	\end{lemma}
	\begin{proof}
		Please see Appendix~\ref{app6}.
	\end{proof}

	With all the Lemmas characterizing NE in the same way as KKT conditions (and individual rationality), we can compare them to prove Theorem \ref{thmain}.
	\begin{proof}[Proof of Theorem \ref{thmain}]
		We know that the four KKT conditions produce a unique solution $ x^{\star} $ (and corresponding $ \lambda^{\star} $).
For the game $ \mathfrak{G} $, if NE exist, then from Lemmas \ref{lemPF}--\ref{lemStat} we can see that at any NE, allocation $ x $ and prices $ \novec{p} $ satisfy the same conditions as the four KKT conditions and hence they give a unique $ x =x^{\star} $, as long as (A5) is satisfied. We conclude that the allocation is $ x^{\star} $ across all NE.
Existence of the claimed NE is established in Lemma~\ref{lemexis}.
		This combined with individual rationality Lemma~\ref{lem_ir}, proves Theorem \ref{thmain}.
	\end{proof}

	\section{A mechanism with Strong Budget Balance} \label{secmechown}
	
	We now present a modification of the previous mechanism that in addition to previous results also ensures \textit{strong budget balance} (SBB) at NE i.e. $ \sum_{ki \in \mN} t_{ki} = 0 $.

	For creating a mechanism in this formulation, the main difference with the previous section, is that the designer has to find a way of redistributing the total tax paid by all the agents. In the last section it was shown that the total payment made at the equilibrium is
	\begin{equation} \label{EQtotpay}
	B = \sum_{ki \in \mN} \left( x_{ki} \sum_{l \in \mL_{ki}} \alpha^l_{ki} p^l_{ki} \right) = r \sum_{ki \in \mN} \left( y_{ki} \sum_{l \in \mL_{ki}} \alpha^l_{ki} p^l_{ki} \right),
	\end{equation}
since all other tax terms are zero at equilibrium. The method here, following in the spirit of \cite{hurwicz1979outcome}, is to redistribute taxes by modifying the tax function for each agent only using messages from other agents. This has the advantage of keeping equilibrium calculations in line with the previous mechanism, 
	since deviations by an agent wouldn't affect his utility through this additional term. In view of this, an alternate expression for $B$ is
	\begin{equation} \label{EQALTREP}
	B = r \sum_{ki \in \mN} \left( \sum_{l \in \mL_{ki}} \frac{1}{N^l-1} \sum_{{k^{\prime} j} \in \mN^l\backslash\{ki\}} \alpha^l_{k^{\prime} j} p_{k^{\prime}j}^l y_{k^{\prime} j} \right),
	\end{equation}
	where each term of the outermost summation depends only on demands of agents other than the $ {ki}\textsuperscript{th} $ one.
	This means that each term in the parenthesis (scaled by the factor $ r $) can now be used as the desired additional tax for user $ ki $.
	Observe however that in our mechanism, each agent's demand affects the factor $ r $ as well.
	So, if all agents can agree on value of $ r $ then that signal can be used to create the term that facilitates budget balance.
	
	In lieu of this, the mechanism here works by asking for an additional signal $ \rho_{ki} $ from every agent and imposing an additional tax of $ (\rho_{ki} - r)^2 $, thereby essentially ensuring that all agents agree on the value of $ r $ (via $ \rho_{ki} $'s) at equilibrium. Finally, $ \bar{\rho}_{-ki} $ (see \eqref{EQRHO}) is used as a proxy for $ r $ in \eqref{EQALTREP} - somewhat similar to how $ \bar{w}_{-k}^l $'s were used in the third term of~\eqref{eq:taxes1}.
	
%
	
	\subsection{Mechanism}
	The actions set for agent $ki$ is now $ \mathcal{S}_{ki} = \mathbb{R}_{+} \times \mathbb{R}_{+}^{2L_{ki}} \times \mathbb{R}_{+} $ where an action is of the form $ s_{ki} = (y_{ki}, Q_{ki}, \rho_{ki}) $. 
	The allocation and tax function announced by the designer are exactly as in the previous case, with the only exception that the tax for any agent $ ki $ is now defined as
	\begin{gather}
	\label{EQt2}
	h^{t}_{ki}(s) = t_{ki} = \zeta (\rho_{ki} - r)^2 + \sum_{l \in \mL_{ki}} t_{ki}^l.
	\end{gather}
	If $ G_k^l \ge 2 $, then again using agents $ kj $ and $ ke $ as described after \eqref{EQt}, we have
	\begin{multline}
	t_{ki}^l = x_{ki} \alpha^l_{ki} {^2q_{kj}^l}  +  ({^2q_{ki}^l} -  {^1q_{ke}^l})^2  + (w_k^l - \bar{w}_{-k}^l)^2
	\\
	+   \eta \, {^2q_{kj}^l} ({^1q_{ki}^l} -  {^2q_{kj}^l})(m_k^l - \alpha_{ki}^l x_{ki}) \\ 
	\label{EQt21}
	{} + {} \xi \, \bar{w}_{-k}^l ( w_k^l - \bar{w}_{-k}^l ) ( c^l - \sum_{k^{\prime} \in \mK^l} m_{k^{\prime}}^l )
	\\
	-   \frac{\bar{\rho}_{-ki}}{N^l-1} \sum_{k^{\prime} j \in \mN^l \backslash \{ki\} } \alpha^l_{k^{\prime} j} {} ^2q_{k^{\prime} j^{\prime}}^l y_{k^{\prime} j}
	\end{multline}
	where agent $ k^{\prime} j^{\prime} $ is \tb{an agent from group $ k^\prime $ on link $ l $ who can be identified by the index $ g_{k^{\prime}}^l(j) - 1 $,} if $ G_{k^{\prime}}^l \ge 2 $ and $ \zeta,\eta,\xi $ are small enough positive constants. However if $ G_{k^{\prime}}^l = 1 $ we would use $ \bar{w}_{-k^{\prime}}^l $ instead of $ q_{k^{\prime} j^{\prime}}^l $.
	
	Similarly for $ G_k^l = 1 $, we have
	\begin{multline}
	t_{ki}^l = x_{ki} \alpha^l_{ki} \bar{w}_{-k}^l + (w_k^l - \bar{w}_{-k}^l)^2
	\\
	+ \eta \, \bar{w}_{-k}^l ({^1q_{ki}^l} -   \bar{w}_{-k}^l)(m_k^l - \alpha_{ki}^l x_{ki})
	\\
	\label{EQt22}
	{} +  \xi \, \bar{w}_{-k}^l (w_k^l - \bar{w}_{-k}^l) (c^l- \sum_{k^{\prime} \in \mK^l} m_{k^{\prime}}^l)
	\\
	-   \frac{\bar{\rho}_{-ki}}{N^l-1} \sum_{k^{\prime} j \in \mN^l \backslash \{ki\} } \alpha^l_{k^{\prime} j} {} ^2q_{k^{\prime} j^{\prime}}^l y_{k^{\prime} j}.
	\end{multline}
	In addition to previous definitions,
	\begin{gather} \label{EQRHO}
	\bar{\rho}_{-ki} \coloneqq \frac{1}{N-1} \sum_{k^{\prime}j \in \mN \backslash \{ki\}} \rho_{k^{\prime}j}.
	\end{gather}
	Denote the corresponding game by $ \mathfrak{G}_0 $.
	The implications of these modifications are discussed in the results section below.
	
	\subsection{Results}
	
	This new mechanism will also fully implement~\eqref{CP}. The only term in $ \widehat{u}_{ki} $ that is affected by $ \rho_{ki} $ is $ -(\rho_{ki} - r)^2 $, so all the Lemmas from Section \ref{secmech} are valid with minor modifications and the main result will follow using the same line of argumentation as for Theorem \ref{thmain}. Note here that, terms in $ \widehat{u}_{ki} $ affected by $ {^1q_{ki}^l}, {^2q_{ki}^l} $'s are the same as before but for $ y_{ki} $ there is a new term $ -(\rho_{ki} - r)^2 $ which is affected by it.
	
	\begin{theorem}[Full Implementation with SBB]\label{thmain2}
		For game $ \mathfrak{G}_0 $, there is a unique allocation, $ x $, corresponding to all NE. Moreover, $ x = x^{\star} $, the maximizer of \eqref{CP}, where individual rationality is satisfied for all agents in $ \mN $. Also, SBB is satisfied at all NE.
	\end{theorem}

	For economy of exposition, in the following, we only provide detailed proofs for the new properties of this mechanism, while we outline the proofs for the properties that are similar to those in the WBB mechanism.
	\begin{itemize}
		\item Primal Feasibility - Since allocation function is the same as before, this result holds here as well.
		
		\item Equal Prices at equilibrium - This was proved by taking price deviations only and keeping other parameters of the signal constant, so the same argument works here as well (noting that no new price related terms have been added in the new mechanism).
	\end{itemize}
	Before moving on to other results, we show that a common $ \rho_{ki} $'s emerges at equilibrium.
	
	\begin{lemma} \label{lemcmr}
		At any NE $ s = (y, P, \rho) $ of game $ \mathfrak{G}_0 $, we have $ \rho_{ki} = r $ $ ~\forall $ $ ki \in \mN $.
	\end{lemma}
	\begin{proof}
		Please see Appendix~\ref{app9}.
	\end{proof}

	Note however that although $ \rho_{ki} $ are same for all $ ki $ at any equilibrium, that common value, $r$,  will be different across equilibria. This is obvious since magnitude of vector $ y $ changes across equilibria.
	
	Now we continue with other properties from Section \ref{secmech}.
	\begin{itemize}
		\item Dual Feasibility - This is obvious here as well.
		
		\item Complimentary Slackness - This was proved by taking only price deviations and hence the same argument works here as well.

		\item Stationarity - Compared to the WBB case, the additional term in the derivative  is
		\begin{equation}
		\frac{\partial \widehat{u}_{ki}}{\partial y_{ki}^{\prime}} \Big \vert_{new} = \underbrace{\frac{\partial \widehat{u}_{ki}}{\partial y_{ki}^{\prime}} \Big \vert_{old}}_{T_1} - \underbrace{2 \zeta (\rho_{ki}^{\prime} - r^{\prime}) \left(-\frac{\partial r^{\prime}}{\partial y_{ki}^{\prime}}\right)}_{T_2}
		\end{equation}
		The claim is, as before, that if $ T_1 $ is positive, agent $ ki $ can increase $ y_{ki}^{\prime} $ from $ y_{ki} $ to be better-off. However agent $ ki $ has to ensure that he deviates with $ \rho_{ki}^{\prime} $ simultaneously to make it equal to $ r^{\prime} $, so that the contribution of the $ T_2 $ term to the derivative above continues to be zero. The only thing left to notice here is that the change in $ \rho_{ki}^{\prime} $ is such that not only the term $ T_2 $ is zero but also that the contribution of term $ - \zeta (\rho_{ki}^{\prime} - r^{\prime})^2 $ to the utility is zero before and after deviation - so this deviation doesn't change other partial derivatives. Similar argument also works when $ T_1 $ is negative and we get the stationarity property here as well.
	\end{itemize}
	With the above properties, unique allocation $ x^{\star} $ at all equilibria is guaranteed, and, as before, the prices are equal to $ \lambda^{\star} $.
	
	Existence of equilibria is verified in Appendix~\ref{app7}. The arguments used are similar to the ones in the proof of Lemma~\ref{lemexis}.

	\begin{itemize}
		\item Individual Rationality - This is obvious in here because money from the previous case is only being redistributed here, so if the mechanism there was individually rational it will be here too.
	\end{itemize}
	
	\begin{lemma}[Strong Budget Balance] \label{lemSBB}
		At any NE $ s = (y,P, \rho) $ of game $ \mathfrak{G}_0$, with corresponding taxes $ \{t_{ki}\}_{ki \in \mN} $, we have $ \sum_{ki \in \mN} t_{ki} = 0 $.
	\end{lemma}
	\begin{proof}
		Please see Appendix~\ref{app8}.
	\end{proof}

	\begin{proof}[Proof of Theorem \ref{thmain2}]
		By the preceding properties, at all equilibria the allocation is $ x^{\star} $ and prices $ \lambda^{\star} $. Then SBB and individual rationality provide the desired full implementation.
	\end{proof}

\section{Discussion and Future work} \label{secgen}
	
	In this paper, we present a mechanism that fully implements the sum of utilities maximizing allocation for agents who share data on a  multicast/multirate network. The scope of application of this model goes beyond just data provision on networks. Another example (as mentioned in the Introduction) is the provision of security products for server farms. The design also encompasses two important auxiliary properties: off-equilibrium feasibility of allocation (using the radial projection function) and SBB at NE. In addition the overall size of the message space is linear in the number of agents, $ N $.

	The work in~\cite{demosmulticorrection} is  similar to the work in this paper except for two things. Firstly, the mechanism doesn't posses the off-equilibrium feasibility property and secondly there are instances of the centralized problem where the claimed NE profiles do indeed have profitable unilateral deviations, thereby contradicting the full implementation claim. 

	In the spirit of some recent works~\cite{chen2002family,healy2012designing} for Walrasian and Lindahl allocations, the main continuation of this work would be to design the allocation/tax functions where in addition to full implementation, convergence of certain classes of learning algorithms is also taken as a design objective.
This way we can convert our ``learning ready'' mechanism into a guaranteed learning mechanism.
Designing so that the ensuing game is a potential game or a super-modular game is too stringent for this problem. But the design technique used in~\cite{healy2012designing}, namely contractive best response correspondence, is slightly relaxed. Basic design parameters like the structure of the message space (consisting of demands and prices) can be kept the same in this case.
	
	\tb{Another future direction, from an engineering perspective, can be to impose additional constraints on the mechanism design so that the messages in the resulting mechanism can be exchanged in a distributed manner and can possibly eliminate the need for a centralized coordinator to collect all quoted message and impose allocation and taxes. Naturally, such a development would require the assumption that even without the presence of a centralized coordinator agents will pay the taxes that the mechanism contract requires and that allocation (as defined by the mechanism) can be effected.}
	
	\tb{For the mechanism defined in this paper the communication structure is as follows. The prices $ {^1q_{ki}^l},{^2q_{ki}^l} $ can be communicated only locally within group $ k $ and link $ l $. The quantity $ \bar{w}_{-k}^l $ requires communication between groups and for this a leader can be designated from each group, who will communicate with the other leaders. Calculating the allocation poses other issues. The leader of group $ k $ on link $ l $ can also calculate $ n_k^l = \max_{i \in \mathcal{G}_k^l } \{ \alpha_{ki}^l y_{ki} \} $ and communicate it with other groups on link $ l $ to calculate $ r^l $ and then communication is needed across links to get $ r = \min_{l \in \mathcal{L}} r^l $.}

\tb{\subsubsection*{Acknowledgment} Authors would like to thank the anonymous reviewers of this paper for their reviews and suggestions for possible future research direcitons.}

	
\appendices

\section{Proof of Lemma~\ref{lemPF}}
\label{app1}

\begin{proof}
	Constraint $ C_1 $ is clearly always satisfied. For $ \novec{y} = \novec{0} $ we have $ x = 0 $ and $ m = 0 $, so constraints $ C_2 $ and $ C_3 $ are also clearly satisfied. Next we show $ C_2 $ and $ C_3 $ are satisfied for any $ y \ne 0 $ as demand. Firstly $ r < + \infty $, since there exists at least one link $ d $ with $ \vert S^d(y) \vert \ge 1 $ and thus $ r^d < + \infty $. Now, for any link $ l $, there are the following two cases. If $ \vert S^l(y) \vert = 0 $, then the allocation for all agents in $ \mN^l $ is zero (along with the corresponding $ m_k^l $'s), so $ C_2 $ and $ C_3 $ for those links is satisfied. If $ \vert S^l(y) \vert \ge 1 $ then
	\begin{gather}
	\sum_{k \in \mK^l} m_k^l = r \sum_{k \in \mK^l} n_k^l \le r^l \sum_{k \in \mK^l} n_k^l \le \frac{c^l}{\sum_{k \in \mK^l} n_k^l} \sum_{k \in \mK^l} n_k^l = c^l
	\end{gather}
	where the first inequality holds because $ r $ is the minimum of all $ r^l $'s. The second inequality is equality if $ \vert S^l(y) \vert \ge 2 $ and is strict only if $ \vert S^l(y) \vert = 1 $ (see second sub-case in \eqref{EQA1}). For $ C_2 $, take any agent $ ki $ and link $ l \in \mL_{ki} $
	\begin{gather}
	\alpha_{ki}^{l} x_{ki} = r \alpha_{ki}^{l} y_{ki} \le  r n_k^l = m_k^l
	\end{gather}
	where the inequality holds because $ n_k^l $ is the maximum over $ \alpha_{ki}^l y_{ki} $'s for all $ i \in \mG_k^l $.
\end{proof}


\section{Proof of Lemma~\ref{lemcmp}}
\label{app2}

\begin{proof}
	First we show the second part of the lemma, so suppose there are agents $ i,e \in \mG_k^l $ as above, for whom $ {^2q_{ki}^l} \ne {^1q_{ke}^l} $. If agent $ ki $ deviates with $ {^2q_{ki}^l}^{\prime} = {^1q_{ke}^l} $ then we can write the difference in agent $ ki $'s utility after and before deviation by just comparing tax for link $ l $ (since allocation and tax for other links don't change)
	\begin{multline}
	\Delta \widehat{u}_{ki} = - ({^2q_{ki}^l}^{\prime} - {^1q_{ke}^l})^2 + ({^2q_{ki}^l} - {^1q_{ke}^l})^2
	\\
	= ({^2q_{ki}^l} - {^1q_{ke}^l})^2 > 0
	\end{multline}
	which means that the deviation was profitable. This gives us the second part of the lemma. (In addition to defining $ {^2q_{ki}^l} = {^1q_{ke}^l} = p_{ke}^l $ when $ G_k^l \ge 2 $, we will also denote $ {^1q_{ki}^l} = w_k^l = p_{ki}^l $ when $ \mG_k^l = \{i\} $).
	
	For the first part, suppose there is a link $ l $ for which $ (w_k^l)_{k \in \mK^l} $ are not all equal, at equilibrium. Clearly then there is a group $ k \in \mK^l $ for which $ w_k^l > \bar{w}_{-k}^l $ (this can be seen from \eqref{EQpm}). We will show that some agent $ i \in \mG_k^l $ can deviate by reducing price $ {^1q_{ki}^l} $ and be strictly better off, thereby contradicting the equilibrium condition. First we take the case when the group $ k $ is such that $ G_k^l \ge 2 $ and then $ G_k^l = 1 $.
	
	Since $ w_k^l > \bar{w}_{-k}^l $ we must have $ w_k^l > 0 $ and since $ w_k^l = \sum_{i \in \mG_k^l} {^1q_{ki}^l} $ there must be an agent $ i \in \mG_k^l $ for whom $ {^1q_{ki}^l} > 0 $. Take deviation by this agent $ ki $ as $ {^1q_{ki}^l}^{\prime} = {^1q_{ki}^l} - \epsilon > 0 $, for which we can write the difference in utility, just as before, as
	\begin{multline} \label{EQdelta}
	\Delta \widehat{u}_{ki}
	= - \epsilon^2  + 2  \epsilon (w_k^l - \bar{w}_{-k}^l) + \eta \,\epsilon p_{ki}^l (m_k^l - \alpha_{ki}^l x_{ki})
	\\
	+ \xi \,\epsilon \bar{w}_{-k}^l (c^l - \sum_{k \in \mK^l} m_k^l) \\
	= \epsilon \Big( -\epsilon +  2 (w_k^l - \bar{w}_{-k}^l) + \eta \, p_{ki}^l (m_k^l - \alpha_{ki}^l x_{ki})
	\\
	+ \xi \, \bar{w}_{-k}^l (c^l - \sum_{k \in \mK^l} m_k^l) \Big) = \epsilon (-\epsilon + a)
	\end{multline}
	where $ a > 0 $ because of Lemma \ref{lemPF} and the fact that $ w_k^l > \bar{w}_{-k}^l $. So by taking $ \epsilon $ such that $ \min\,\{a, {^1q_{ki}^l}\} > \epsilon > 0 $, the above deviation will be a profitable one for agent $ ki $. This gives the result for $ G_k^l \ge 2 $.
	
	For $ G_k^l = 1 $, say $ \mG_k^l = \{i\} $, we have that $ {^1q_{ki}^l} = w_k^l > \bar{w}_{-k}^l $. This again means that $ {^1q_{ki}^l} > 0 $ and we take the deviation $ {^1q_{ki}^l}^{\prime} = {^1q_{ki}^l} - \epsilon > 0 $ and get
	\begin{multline}
	\Delta \widehat{u}_{ki} = \epsilon \Big( -\epsilon + 2 (w_k^l - \bar{w}_{-k}^l) + \eta \, \bar{w}_{-k}^l (m_k^l - \alpha_{ki}^l x_{ki})
	\\
	+ \xi \, \bar{w}_{-k}^l (c^l - \sum_{k \in \mK^l} m_k^l)  \Big).
	\end{multline}
	Following the same argument as above we will get our result here as well.
\end{proof}


\section{Proof of Lemma~\ref{lemCompSlac}}
\label{app3}

\begin{proof}
	Suppose there is a link $ l $ for which $ w^l > 0 $ and $ \sum_{k \in \mK^l} m_k^l < c^l $. Take any group $ k \in \mK^l $ and an agent $ i \in \mG_k^l $ such that $ {^1q_{ki}^l} = p_{ki}^l > 0 $ (there is such an agent because $ w^l = \sum_{i \in \mG_k^l} {^1q_{ki}^l} > 0 $). Take the deviation $ {^1q_{ki}^l}^{\prime} = {^1q_{k,i}^l} - \epsilon > 0 $ and we get (using same arguments as in \eqref{EQdelta} and noting that $ w_k^l = \bar{w}_{-k}^l = w^l $)
	\begin{multline}
	\Delta \widehat{u}_{ki} = \epsilon \Big( -\epsilon + \eta \underbrace{p_{ki}^l (m_k^l - \alpha_{ki}^l x_{ki})}_{\ge 0 \text{ by Lemma \ref{lemPF}}} +  \xi \, w^l (c^l - \sum_{k \in \mK^l} m_k^l) \Big)
	\\
	= \epsilon (-\epsilon + a).
	\end{multline}
	where $ a > 0 $ due to Lemma \ref{lemPF} and the assumption that $ w^l (c^l - \sum_{k \in \mK^l} m_k^l)  > 0 $. This gives us that $ w^l (c^l - \sum_{k \in \mK^l} m_k^l)  = 0 $ for all $ l \in \mL $ at equilibrium.
	
	Now suppose there is an agent $ ki $ for whom $ {^1q_{ki}^l} = p_{ki}^l > 0 $ and $ \alpha_{ki}^l x_{ki} < m_k^l $. Same as before, we take the deviation $ {^1q_{ki}^l}^{\prime} = {^1q_{ki}^l} - \epsilon > 0 $,
	\begin{gather}
	\Delta \widehat{u}_{ki} = \epsilon \left( -\epsilon + \eta \, p_{ki}^l (m_k^l - \alpha_{ki}^l x_{ki}) \right)  = \epsilon (-\epsilon  + a)
	\end{gather}
	where $ a > 0 $ by assumption. This gives us that $ p_{ki}^l (m_k^l - \alpha_{ki}^l x_{ki}) = 0 $ for all $ ki \in \mN^l $ and $ l \in \mL $, at equilibrium.
\end{proof}


\section{Proof of Lemma~\ref{lemStat}}
\label{app4}

\begin{proof}
	\eqref{EQS2} is true by construction since we defined $ w_k^l = \sum_{i \in \mG_k^l} p_{ki}^l $ and by Lemma \ref{lemcmp} we have $ p_{ki}^l = {^1q_{ki}^l} $ and $ w_k^l = w^l $.
	
	At any NE, agent $ ki $'s utility in the game $ \widehat{u}_{ki}(s^{\prime}) = v_{ki}(h^{x}_{ki}(s^{\prime})) - h^{t}_{ki}(s^{\prime}) $ as a function of his message $ s_{ki}^{\prime} = (y_{ki}^{\prime}, Q_{ki}^{\prime}) $, with $ s_{-{ki}} $ fixed, should have a global maximum at $ {s}_{ki} = ({y}_{ki}, P_{ki}) $.
	This would mean that if this function was differentiable w.r.t. $y_{ki}^{\prime}$ at $s$, the partial derivatives w.r.t. $ y_{ki}^{\prime} $ at $ s $ should be $ 0 $.
	However, since our allocation dilates/shrinks demand vector $ \novec{y}^{\prime} $ on to the feasible region, it could be the case that increasing and decreasing $ y_{ki}^{\prime} $ gives allocations lying on different hyperplanes, meaning that the transformation from $ \novec{y}^{\prime} $ to $ {x}^{\prime} $ is different on both sides of $ {y}_{ki} $ and therefore we conclude that $ \widehat{u}_{ki} $ could be non-differentiable w.r.t $ y_{ki}^{\prime} $ at $ s $.
	Important thing here however is to notice that right and left derivatives exist, it's just that they may not be equal.  Hence we can take derivatives on both sides of $ {y}_{ki} $ as (noting that derivative of the other terms in utility involving $ x_{ki} $ or involving $ m_k^l $ will be zero due Lemma \ref{lemcmp})
	\begin{subequations}
		\label{EQder}
	\begin{gather}
		\label{EQderA}
	\frac{\partial \widehat{u}_{ki}}{\partial y_{ki}}^{\prime} \Big \vert_{y_{ki}^{\prime} \downarrow {y}_{ki}} = \left( v_{ki}^{\prime}({x}_{ki}) - \sum_{l \in \mL_{ki}} p_{ki}^l \alpha^l_{ki}  \right) \frac{\partial x_{ki}^{\prime}}{\partial y_{ki}^{\prime}} \Big \vert_{y_{ki}^{\prime} \downarrow {y}_{ki} }
	\\[0.1cm]
		\label{EQderB}
	\frac{\partial \widehat{u}_{ki}}{\partial y_{ki}}^{\prime} \Big \vert_{y_{ki}^{\prime} \uparrow {y}_{ki}} = \left( v_{ki}^{\prime}({x}_{ki}) - \sum_{l \in \mL_{ki}} p_{ki}^l \alpha^l_{ki}  \right) \frac{\partial x_{ki}^{\prime}}{\partial y_{ki}^{\prime}} \Big \vert_{y_{ki}^{\prime} \uparrow {y}_{ki} }
	\end{gather}
	\end{subequations}
	We will first show that the $ \partial x_{ki}/\partial y_{ki} $ term above (for either equation) is always positive. If $ y = 0 $ then clearly this is true, because if any agent $ ki $ demands $ y_{ki} > \epsilon $ while $ y_{-ki} = 0 $ then clearly $ x_{ki} > 0 $ (in fact the allocation is differentiable at $ y = 0 $). If $ y \ne 0 $ from \eqref{EQx}, we can write
	\begin{equation}
	\beta \coloneqq \frac{\partial x_{ki}}{\partial y_{ki}} = \frac{\partial (r y_{ki})}{\partial y_{ki}} = r + y_{ki} \frac{\partial r}{\partial y_{ki}} = r^q + y_{ki} \frac{\partial r^q}{\partial y_{ki}}
	\end{equation}
	where $ r = r^q $.
	
	From here we divide our arguments into following cases: (A) $ ki \notin \mN^q $ ($ \Leftrightarrow $ $ i \notin \mG_k^q $); (B) $ ki \in \mN^q $, $ i \notin \arg\max_{j \in \mG_k^q}\{ \alpha_{kj}^q y_{kj }\} $ and (C) $ ki \in \mN^q $, $ i \in \arg\max_{j \in \mG_k^q}\{ \alpha_{kj}^q y_{kj }\} $.
	
	(A) For this clearly $ \partial r^q / \partial y_{ki} = 0 $ and this makes $ \beta = r^q > 0 $.
	
	(B) Since value of $ r^q $ depends only on the value of $ (n_k^q)_{k \in \mK^q} $, and in this case changes in $ y_{ki} $ don't affect $ n_k^q $ we can see that $ \partial r^q / \partial y_{ki} = 0 $ and so $ \beta = r^q > 0 $.
	
	(C) We divide this case into two cases: $ \vert S^q(y) \vert \ge 2 $ or $ \vert S^q(y) \vert = 1 $. If $ \vert S^q(y) \vert \ge 2 $ then
	\begin{equation}
	\beta = r^q + y_{ki} \frac{\partial r^q}{\partial y_{ki}} = r^q + y_{ki} \left( - \frac{c^q}{(\sum_{k_0 \in \mK^q} n_{k_0}^q)^2} \right) \frac{\partial n_k^q}{\partial y_{ki}}.
	\end{equation}
	Now $ \frac{\partial n_k^q}{\partial y_{ki}} $ is either $ \alpha_{ki}^q $ or $ 0 $. If it is $ 0 $ then $ \beta = r^q > 0 $. Otherwise we have
	\begin{equation}
	\beta = \frac{(r^q)^2}{c^q} \sum_{k_0 \in \mK^q \backslash \{k\} } n_{k_0}^q
	\end{equation}
	which is positive because $ \vert S^q (y) \vert \ge 2 $, since then there is at least one positive term in the summation. For $ \vert S^q(y) \vert = 1 $, we will consider $ S^q(y) = \{k\} $; else in case $ S^q(y) = \{k_0\} \ne \{k\} $, taking the derivative would give the same expression as above. For $ S^q(y) = \{k\} $ we will get
	\begin{subequations}
	\begin{gather}
	r^q = \frac{c^q}{n_k^q} - \frac{c^q}{n_k^q(n_k^q + 1)} = \frac{c^q}{n_k^q} \left( 1 - \frac{1}{n_k^q + 1} \right) = \frac{c^q}{n_k^q + 1}
	\\
	\beta = r^q + y_{ki} \, \frac{\partial r^q}{\partial n_k^q}\, \frac{\partial n_k^q}{\partial y_{ki}}
	\end{gather}
	\end{subequations}
	where $ \frac{\partial n_k^q}{\partial y_{ki}} $ is either $ 0 $ or $ \alpha_{ki}^q $. If it is $ 0 $ then $ \beta = r^q > 0 $ and if it is $ \alpha_{ki}^q $ then we have $
	\beta = \frac{(r^q)^2}{c^q} $.
	Hence we have $ \beta > 0 $ in all cases.
	
	Referring to \eqref{EQder}, there are two possibilities, the first term on RHS in both equations in \eqref{EQder} is positive or negative. If it's positive, then we can see from \eqref{EQderA} that by increasing $ y_{ki}^{\prime} $ from $ {y}_{ki} $ (and therefore $ x_{ki}^{\prime} $ from $ {x}_{ki} $) agent $ ki $ can increase his pay-off, which contradicts equilibrium. Now similarly consider the first term in RHS of \eqref{EQder} to be negative, then from \eqref{EQderB}, agent $ ki $ can reduce $ y_{ki}^{\prime} $  from $ y_{ki} $ to get a better pay-off. But the downward deviation in $ y_{ki}^{\prime} $ is only possible if $ {y}_{ki} > 0 $ ($ \Leftrightarrow~ {x}_{ki} > 0 $). Hence we conclude~\eqref{EQKKTst1}, \eqref{EQKKTst2} from the statement of this lemma.
\end{proof}


\section{Proof of Lemma~\ref{lemexis}}
\label{app5}

\begin{proof}
We will show that $ s = (y,Q) $ is a NE of the game $ \mathfrak{G} $, where for each agent $ ki \in \mN $, $ y_{ki} = \rho x_{ki}^\star $ (for any $ \rho > 0 $) and
\begin{equation}\label{EQcandNE}
{^1q_{ki}^l} = {\mu_{ki}^l}^\star \quad \text{and}  \quad  {^2q_{ki}^l} = {\mu_{ke}^l}^\star,
\end{equation}
where agent $ ke $ is the one identified by index $ g_k^l(i)+1 $. Here $ x^\star,\mu^\star $ are the primal-dual variables that satisfy the KKT conditions of~\eqref{CP}.
	
Define the function $f: \mathcal{S}_{ki} \rightarrow \mathbb{R}$  as follows
\begin{subequations}
\begin{align}	
f(s_{ki}^\prime) &\triangleq \widehat{u}_{ki}(s_{ki}^\prime,s_{-ki}) \\
&= v_{ki}\big(h^x_{ki}(s_{ki}^\prime,s_{-ki})\big) - h^t_{ki}(s_{ki}^\prime,s_{-ki}).
\end{align}
\end{subequations}
The above claim will be shown in two steps: In Step~1 we show that message $ s $ results in allocation $ x^\star $ and additionally we have $ \nabla f(s_{ki}) = 0 $. In Step~2 we show that for any agent $ ki $, unilateral deviations from $ s_{ki} $ are not strictly profitable.

\textbf{Step 1.} Assumption (A4) implies that with $ y = \rho x^\star $, while calculating allocation $ \big( h_{ki}^x(s) \big)_{ ki \in \mN } $, sub-case 1 in eq.~\eqref{EQA1} will be valid. This is turn implies that $ h^x(s) = x^\star $. Thus, message $ y = \rho x^\star $ produces allocation $ x^\star $. Now, using the complimentary slackness property of KKT and the fact that in the claimed NE above, we have for any $ ki \in \mN $, $ {^2q_{ki}^l} = {^1q_{ke}^l} = {\mu_{ke}^l}^\star $, it can be easily shown that the partial derivative of $ f $ w.r.t. $ ^1q_{ki}^l $ evaluated at the point $ s_{ki} $  is zero and the same is true for the partial derivative w.r.t $ ^2q_{ki}^l $. Finally using the stationarity condition in KKT we deduce that the partial derivative of $ f $ w.r.t. $ y_{ki} $ evaluated at the point $ s_{ki} $ is zero. Thus, $ \nabla f(s_{ki}) = 0 $.

\textbf{Step 2.}	We now check for profitable deviations. This task is accomplished in two steps.
In Step 2a, we show that any interior local extrema $ \tilde{s}_{ki} = \big( \tilde{y}_{ki}, {^1\tilde{q}_{ki}^l}, {^2\tilde{q}_{ki}^l} \big) $ of $ f(\cdot) $ satisfies $ \nabla f(\tilde{s}_{ki}) = 0 $. Furthermore, we show that any interior local extrema $ \tilde{s}_{ki}$ produces allocation $ h^x(\tilde{s}_{ki},s_{-ki}) $ and prices that satisfy the KKT conditions as $ x^\star,\mu^\star $.
In Step 2b, using Step 2a, we show that all interior local extrema are local maxima.
As a result of this and using the fact that by construction, $ f(s_{ki}^\prime) $ is continuous w.r.t. $ s_{ki}^\prime $, we conclude that $ f(\cdot) $ can have only a single local maximum and this is also the  global maximum. This is because all local extrema are necessarily characterized by $ \nabla f(\tilde{s}_{ki}) = 0 $ (Step 2a) and there cannot be multiple local maxima in the interior without a local minima between them. Furthermore, using this argument we know that this global  maximizer has to be in the interior. Finally, using the fact that $ \nabla f(s_{ki}) = 0 $ (Step 1) we conclude that this maximizer is $ s_{ki} $. In other words, agent $ ki $ cannot be strictly better off by unilaterally deviating from $ s_{ki} $. As this is checked for any agent $ ki \in \mN $, we conclude that above mentioned message $ s=(y,Q) $ is a NE.

\textbf{Step 2a.} First we will show that without the gradient being zero, there cannot be a local extremum. Note that by construction, $ f(s_{ki}^\prime) $ is continuous w.r.t. $ s_{ki}^\prime  $.	 Since $ f(\cdot) $ is differentiable w.r.t. $ {^1q_{ki}^l}, {^2q_{ki}^l} $ variables (at all points), derivatives w.r.t. these indeed have to be $ 0 $ at any interior local extremum. This implies equal prices and complimentary slackness properties (using arguments from respective proofs) and for the remaining term we can write
	\begin{gather}
	\frac{\partial f}{\partial y_{ki}} = \left(v_{ki}^{\prime}(x_{ki}) - \sum_{l \in \mL_{ki}} \alpha^l_{ki} {p_{ki}^l}\right) \left(\frac{\partial x_{ki}}{\partial y_{ki}}\right).
	\end{gather}
	Note that in the proof of Lemma~\ref{lemStat}, we have shown that $ \beta \coloneqq \frac{\partial x_{ki}}{\partial y_{ki}} > 0 $ always. So at the points of non-differentiability, $ \beta $ will have a jump discontinuity, however it will be positive on either side. It is then clear that without satisfying $ v_{ki}^{\prime}(x_{ki}) = \sum_{l \in \mL_{ki}} \alpha^l_{ki} {p_{ki}^l} $, there cannot be a local extremum. This implies that $ \nabla f(\tilde{s}_{ki}) = 0 $ at any interior local extremum. Furthermore, KKT conditions must be satisfied at the corresponding allocation $ h^x(\tilde{s}_{ki},s_{-ki}) $ and prices $ \mu^\star,\lambda^\star $. For this note that $ {^1\tilde{q}_{ki}^l} = {^2q_{kj}^l} $ (where $ kj $ is the agent with index $ g_k^l(i) - 1 $) and from~\eqref{EQcandNE} we know that $ {^2q_{kj}^l} = {\mu_{ki}^l}^\star $. Thus $ {^1\tilde{q}_{ki}^l} = {\mu_{ki}^l}^\star $ and similarly $ {^2\tilde{q}_{ki}^l} = {\mu_{ke}^l}^\star $ (where $ ke $ is the agent with index $ g_k^l(i) + 1 $). This combined with $ v_{ki}^{\prime}(x_{ki}) = \sum_{l \in \mL_{ki}} \alpha^l_{ki} {p_{ki}^l} $ from above, implies that KKT is satisfied at the local extrema.

\textbf{Step 2b.} The Hessian $ H $ of $ f(\cdot) $ w.r.t. $ s_{ki}^\prime $ is of size $ (2L_{ki}+1) \times (2L_{ki}+1) $, where the 1st row and column represent $ y_{ki}^\prime  $ and the two subsequent sets of $ L_{ki} $ rows and columns represent $ \{{^1q_{ki}^l}^\prime \}_{l \in \mL_{ki}} $ and $ \{{^2q_{ki}^l}^\prime \}_{l \in \mL_{ki}} $. We want $ H $ (evaluated at any local extremum) to be negative definite.
Looking at the diagonal entries, and using what we have derived in the previous paragraph, it can be calculated that the only non-zero entries at any local extrema are negative
$ f_{{^1q} {^1q}} \coloneqq \frac{\partial^2 f}{\partial {^1q_{ki}^{l}} \partial {^1q_{ki}^{l}} } = -2 $,
$ f_{{^2q}{^2q}} \coloneqq \frac{\partial^2 f}{\partial {^2q_{ki}^{l}}\partial {^2q_{ki}^{l}} } = -2 $,
$ f_{yy} \coloneqq \frac{\partial^2 f}{\partial y_{ki}^2} = v_{ki}^{\prime \prime}(x_{ki}) \left(\frac{\partial x_{ki}}{\partial y_{ki}}\right)^2 $ (due to strict concavity of $ v_{ki} $).
Also notice that all off-diagonal entries, except the first $ L_{ki}+1 $ in the first row and column, are zero.
Finally, note that due to assumption (A2), all prices are finite at local extremum and so the aforementioned non-zero entries
\begin{multline}
f_{{^1q}y} \coloneqq \frac{\partial^2 f}{\partial {^1q_{ki}^l} \partial y_{ki}} = \eta \, {^2q_{kj}^l} \left(\alpha_{ki}^l \frac{\partial x_{ki}}{\partial y_{ki}} - \frac{\partial m_k^l}{\partial y_{ki}} \right)
\\
{}+{} \xi \, \bar{w}_{-k}^l \left(\sum_{k^{\prime} \in \mK^l} \frac{\partial m_{k^{\prime}}^l }{\partial y_{ki}} \right)
\end{multline}
are finite.
We now show that the roots of the characteristic polynomial of $ H $ (i.e. its eigenvalues) all become negative for $ \eta,\xi $ chosen to be sufficiently small.

For this, we take a generic matrix $ A_0 $, which is similar in structure to $ H $ and whose entries have the same dependence on $ \vert y \vert $ as $ H $. So $ A_0 $ is of the form
	$
	A_0 = \begin{bmatrix}
	A & 0\\
	0 & D
	\end{bmatrix}
	$
	where $ D = (-2) I_{L_{ki}} $. In this case we know that eigenvalues of $ A $ and $ D $ together will give us all the eigenvalues of $ A_0 $. Clearly eigenvalues of $ D $ are $ -2 $ repeated $ L_{ki} $ times, so all that we now need to do is check whether all eigenvalues of $ A $ are negative. Entries in $ A $ are $ a_{11} = -\frac{a}{\vert y \vert^2} $, $ a_{ij} = a_{ji} = 0 $ $ \forall~ i,j > 1,~ i \ne j $ and $ a_{ii} = -2 $, $  a_{1i} = a_{i1} = \eta  \frac{b_{i-1}}{\vert y \vert} $ $ \forall~ 2 \le i \le L_{ki}+1 $.
where $ a > 0 $ (and we don't care about the sign of $ b_i $'s). The factor of $ \eta $ in front of $ a_{i1},a_{1i} $ is to be taken as $ \max\left( \eta,\xi \right) $, but since we can set either of them we take it simply as $ \eta $ here. The parameters $ a,b_i $ may not be completely independent of $ y $ but since the absolute value of $ y $ has been taken out of the scaling, their values are bounded. Magnitude of $ b_i $'s are bounded from above and $ a $ is bounded away from zero.
	
	We can explicitly calculate $ \vert A - \lambda I \vert $ and write the characteristic equation as
	\begin{multline}
	Q(\lambda) = \left(-\frac{a}{\vert y \vert^2} - \lambda \right) (-2-\lambda)^{L_{ki}}
	\\
	{} + {} \eta^2 \frac{\sum_{i = 1}^{L_{ki}} (-1)^{i} b_i^2}{\vert y \vert^2} (-2-\lambda)^{L_{ki} - 1}  = 0.
	\end{multline}
	So $ -2 $ is a repeated eigenvalue, $ L_{ki} - 1 $ times. The equation for the remaining two roots can be written as
	\begin{equation}
	\left(-\frac{a}{\vert y \vert^2} - \lambda \right) (-2-\lambda) + \eta^2  \frac{C}{\vert y \vert^2}  = 0.
	\end{equation}
	Necessary and sufficient conditions for both roots of this quadratic to be negative are
	\begin{equation}
	\left(2 + \frac{a}{\vert y \vert^2}\right) > 0, \qquad \frac{2a}{\vert y \vert^2} + \eta^2 \frac{C}{\vert y \vert^2} > 0,
	\end{equation}
	the first of which is always true, since $ a > 0 $. The second one can be ensured by making $ \eta $ small enough, since $ a $ is bounded away from zero and the magnitude of $ C $ is bounded from above.
	
	Hence the Hessian $ H $ is shown to be negative definite for $ \eta $ chosen to be small enough.
\end{proof}


\section{Proof of Lemma~\ref{lem_ir}}
\label{app6}

\begin{proof}
	Because of Lemma \ref{lemcmp}, the only non-zero term in $ t_{ki} $ (see \eqref{EQt}) at equilibrium is $ x_{ki} \sum_{l \in \mL_{ki}} \alpha^l_{ki} p^l_{ki} $, which is clearly non-negative. Hence $ \sum_{ki \in \mN} t_{ki} \ge 0 $ at equilibrium. 

	Now if  $ x_{ki} = 0$ then we know from Lemma \ref{lemcmp} and  \eqref{EQt} that $ t_{ki} = 0 $ and so \eqref{EQIRB} is evident. Now take $ x_{ki} > 0 $ and define the function
	\begin{equation}
	f(z) = v_{ki}(z) - z \sum_{l \in \mL_{ki}} \alpha^l_{ki} p^l_{ki}.
	\end{equation}
	Note that $ f(0) = u_{ki}(0,0) $ and $ f(x_{ki}) = u_{ki}(x,t) $, the utility at equilibrium. Since $ f^{\prime}(x_{ki}) = 0 $ (Lemma \ref{lemStat}), we see that $ \forall $ $ 0 < y < x_{ki} $, $ f^{\prime}(y) > 0 $ since $ f $ strictly concave (because of $ v_{ki} $). This clearly implies $ f(x_{ki}) \ge f(0) $.
\end{proof}


\section{Proof of Lemma~\ref{lemcmr}}
\label{app9}

	\begin{proof}
		Suppose not, i.e. assume $ \exists $ $ kj \in \mN $ such that $ \rho_{kj} \ne r $. In this case agent $ kj $ can deviate with only changing $ \rho_{kj}^{\prime} = r $ (which also means $ r $ is the same as before deviation, since demand $ y $ doesn't change). It's easy to see that this is a profitable deviation, since change in utility of agent $ kj $ will be only through the term involving $ \rho_{kj} $.
		\begin{equation}
		\Delta \widehat{u}_{kj} = - \zeta (\rho_{kj}^{\prime} - r)^2 + \zeta (\rho_{kj} - r)^2  = \zeta (\rho_{kj} - r)^2 > 0.
		\end{equation}
	\end{proof}

%

\section{Proof of Existence in Section~\ref{secmechown}}
\label{app7}

\begin{proof}
	First order conditions can again be shown to be satisfied, the only difference is that here we will also use $ \rho_{ki} = r $ at local extremum. The Hessian $ H $ for any agent $ ki $ here will be of order $ (2L_{ki}+2) \times (2L_{ki}+2) $ where 1st, 2nd row and column represent $ y_{ki} $, $ \rho_{ki} $ respectively whereas the remaining rows and columns represent $ {^1q_{ki}^l} $'s and $ {^2q_{ki}^l} $'s. The generic matrix $ A_0 $ for $ H $ will then be of the form
	$
	A_0 = \begin{bmatrix}
	A & 0 \\
	0 & D
	\end{bmatrix}
	$
	where $ D = (-2) I_{L_{ki}} $ and matrix $ A $, of order $ (L_{ki} + 2) \times (L_{ki} + 2) $, has elements $ a_{11} = -\frac{a}{\vert y \vert^2} - \zeta \frac{d}{\vert y \vert^4} $, $ a_{12} = a_{21} = - \zeta \frac{e}{\vert y \vert^2} $, $  a_{ij} = a_{ji} = 0 $ $ \forall~ i,j > 1,~ i \ne j $, $ a_{22} = -2 $, $ a_{ii} = -2 $, and $ a_{1i} = a_{i1} = \eta \frac{b_{i-1}}{\vert y \vert} $ $ \forall ~3 \le i \le L_{ki}+2 $.
	where $ a,d,e > 0 $. As before, it will be shown that all eigenvalues of $ A $ are negative (since that is clearly true for $ D $). Writing the characteristic equation, it is the case again that $ -2 $ is a repeated eigenvalue, $ L_{ki} $ times. And the equation for remaining two roots is
	\begin{multline}
	\lambda^2 + \lambda \left(2 + \frac{a}{\vert y \vert^2} +  \zeta \frac{d}{\vert y \vert^4}\right)
	\\
	{} + {} \left(\frac{2a}{\vert y \vert^2} + \zeta \frac{2d}{\vert y \vert^4} - \zeta^2 \frac{E}{\vert y \vert^4} + \eta^2 \frac{C}{\vert y \vert^2} \right) = 0
	\end{multline}
	Necessary and sufficient conditions for the roots of above quadratic to be negative are again that coefficient of $ \lambda $ and the constant term are both positive. Coefficient of $ \lambda $ is clearly positive, and the constant term can also be made positive by choosing $ \zeta,\eta $ (\& $ \xi $) small enough, irrespective of sign of $ C $. 
\end{proof}


\section{Proof of Lemma~\ref{lemSBB}}
\label{app8}

\begin{proof}
	Terms 2, 3, 4 and 5 in  \eqref{EQt21} and \eqref{EQt22} are zero at equilibrium and so we can write (at equilibrium)
	\begin{multline}
	\sum_{ki\in\mN} t_{ki} = \sum_{ki\in\mN}  {x}_{ki} \left (\sum_{l \in \mL_{ki}} \alpha^l_{ki} p_{ki}^l \right)
	\\
	{} - {} r \sum_{l \in \mL_{ki}} \frac{1}{N^l - 1} \sum_{k^{\prime} j^{\prime} \in \mN^l \backslash \{ki\} } \alpha^l_{k^{\prime} j^{\prime}} p_{k^{\prime}j^{\prime}}^l y_{k^{\prime} j^{\prime} }
	\\
	\label{EQsbb}
	= \sum_{ki\in\mN} \sum_{l\in\mL_{ki}} \left( {x}_{ki} \alpha^l_{ki} p_{ki}^l - \frac{1}{N^l - 1} \!\!\!\!\!\! \sum_{k^{\prime} j^{\prime} \in \mN^l \backslash \{ki\} } \!\!\!\!\!\!	 \alpha^l_{k^{\prime} j^{\prime}} p_{k^{\prime}j^{\prime}}^l x_{k^{\prime} j^{\prime} } \right).
	\end{multline}
	The coefficient of $ {x}_{ki} $ for any agent $ ki $ in the above is
	\begin{multline}
	\sum_{l\in\mL_{ki}} \left( \alpha^l_{ki} p_{ki}^l - \frac{1}{N^l-1} \sum_{k^{\prime} j^{\prime} \in \mN^l \backslash \{ki\} }  \alpha^l_{ki} p_{ki}^l  \right) = 0
	\end{multline}
	which proves the claim.
\end{proof}

		\bibliographystyle{IEEEtran}
		\bibliography{IEEEabrv,achilleas14abrv,abhinav}
		
\newpage 	

\begin{IEEEbiography}[{\includegraphics[width=1in,height=1.25in,clip,keepaspectratio]{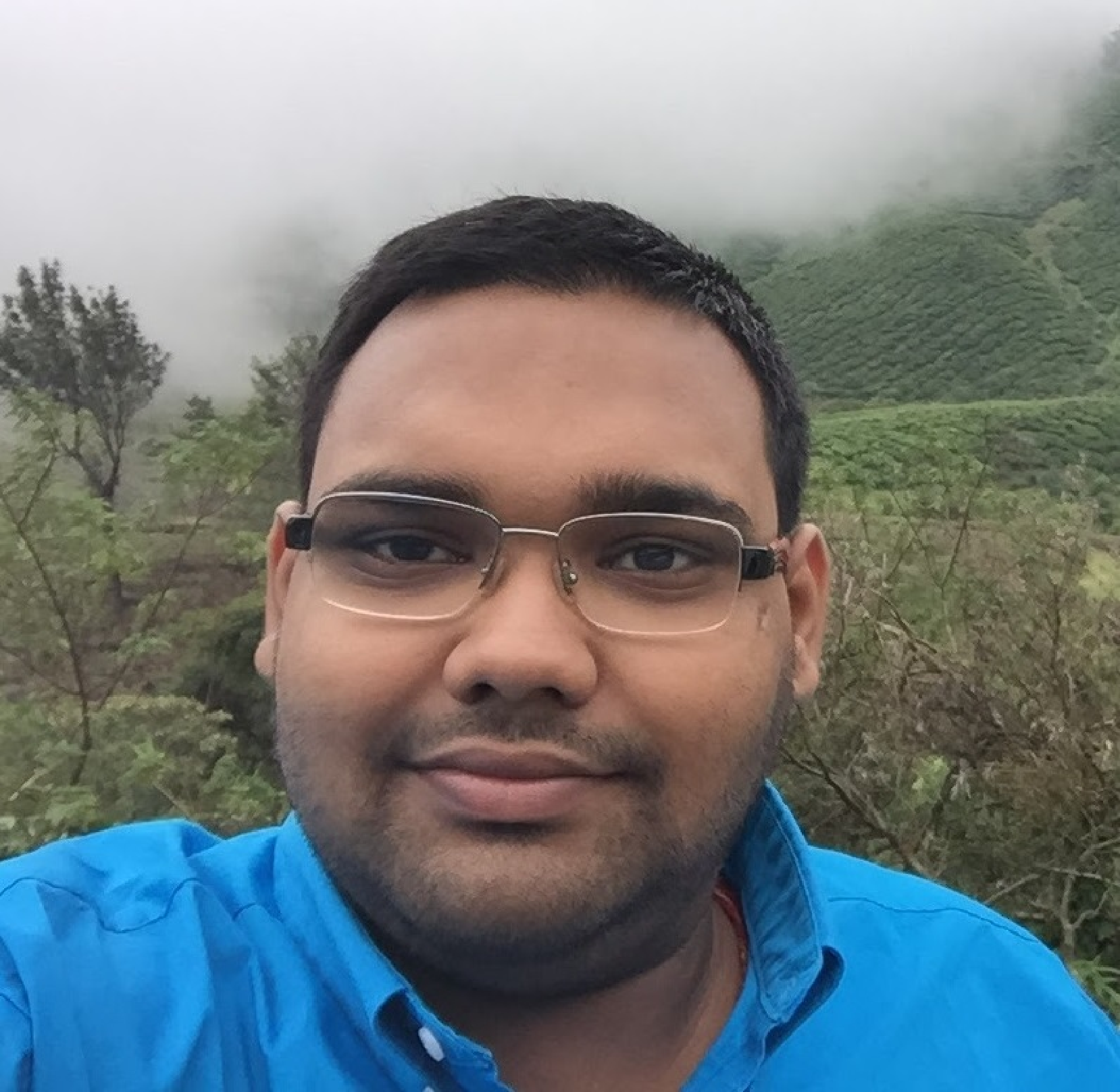}}]{Abhinav Sinha}
	was born in Mumbai, India in 1990. He received B.Tech (with honours) and M.Tech in Electrical Engineering from the Indian Institute of Technology, Bombay in 2012, and M.S.  in Mathematics from the University of Michigan in 2014. He is currently a PhD Candidate at the University of Michigan, Ann Arbor, Department of Electrical Engineering and Computer Science.
	
	His research interests lie in the general area of network communication, with emphasis in economics of networks -- analysis of dynamic games and mechanism design for resource allocation on networked systems; relation between perturbation and regularization based online learning algorithms. 
\end{IEEEbiography}

\begin{IEEEbiography}[{\includegraphics[width=1in,height=1.25in,clip,keepaspectratio]{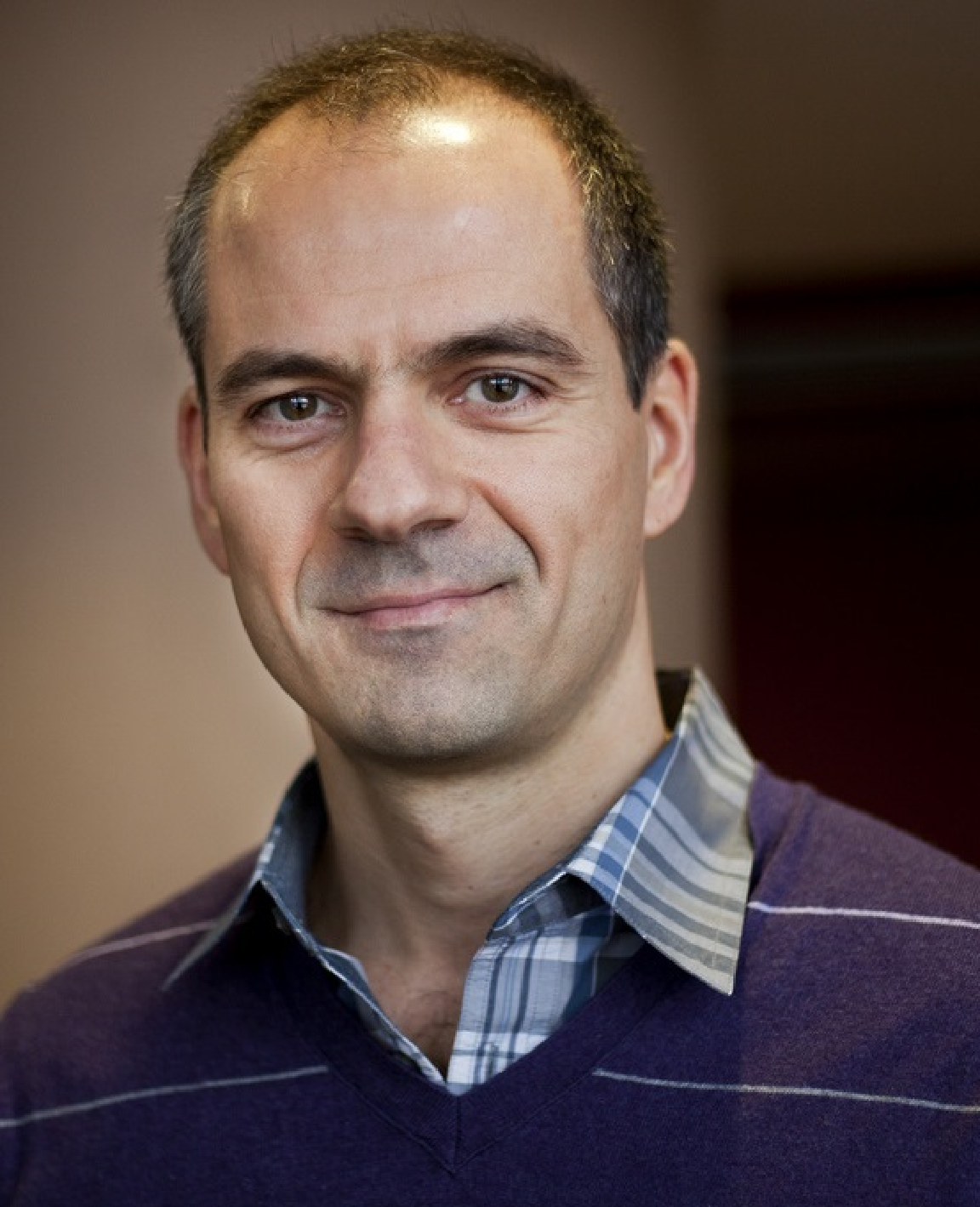}}]{Achilleas Anastasopoulos}
	(S'97-M'99-SM'13) was born in Athens, Greece in 1971. He received the Diploma
	in Electrical Engineering from the National Technical University of
	Athens, Greece in 1993, and the M.S.\ and Ph.D.\ degrees in Electrical
	Engineering from University of Southern California in 1994 and 1999,
	respectively. He is currently an Associate Professor at the University of
	Michigan, Ann Arbor, Department of Electrical Engineering and Computer
	Science.
	
	His research interests lie in the general area of communication and information theory,
	with emphasis in channel coding and multi-user channels;
	control theory with emphasis in decentralized stochastic control and its
	connections to communications and information theoretic problems;
	analysis of dynamic games and mechanism design for resource allocation on networked systems.
	
	He is the co-author
	of the book \emph{Iterative Detection: Adaptivity, Complexity Reduction,
		and Applications,} (Reading, MA: Kluwer Academic, 2001).
	
	Dr.\ Anastasopoulos is the recipient of the ``Myronis Fellowship'' in
	1996 from the Graduate School at the University of Southern California,
	the NSF CAREER Award in 2004,
	and was a co-author for the paper that received the best student paper award in ISIT 2009.
	He served as a technical program committee
	member for ICC 2003, 2015, 2016; Globecom 2004, 2012; VTC 2007, 2014, 2015; ISIT 2015, and on the editorial
	board of the IEEE TRANSACTIONS ON COMMUNICATIONS.
\end{IEEEbiography}
\vfill

\end{document}